\PassOptionsToPackage{dvipsnames}{xcolor}
\documentclass[runningheads, envcountsame]{llncs}
\usepackage[T1]{fontenc}
\usepackage[utf8]{inputenc}
\usepackage{bussproofs}
\usepackage{tikz}
\usepackage{tikz-cd}
\usepackage{mathrsfs}
\usepackage{amsmath}
\usepackage{amssymb}
\usepackage{stmaryrd}
\usepackage{mathtools}
\usepackage{yfonts}
\usepackage{hyperref}

\usepackage{caption}
\usepackage{subcaption}

\usepackage[b]{esvect}

\renewcommand{\vec}[1]{\vv{#1}}

\DeclareMathAlphabet{\mathpzc}{OT1}{pzc}{m}{it}

\spnewtheorem{assumption}[theorem]{Assumption}{\bfseries}{\itshape}
\spnewtheorem{notation}[theorem]{Notation}{\bfseries}{\itshape}

\newenvironment{bprooftree}
  {\leavevmode\hbox\bgroup}
  {\DisplayProof\egroup}

\newcommand{\secref}[1]{\S\ref{#1}}%
\newcommand\Seq{\mathrm{Seq}}%
\newcommand\TerSeq{\mathrm{TerSeq}}%
\newcommand\End{\mathrm{End}}%

\newcommand\inlinesum{\displaystyle \sum}

\newcommand\bra[1]{{\langle{#1}|}}
\newcommand\ket[1]{{|{#1}\rangle}}
\newcommand\trace[1]{\ensuremath{\mathrm{tr}(#1)}}

\newcommand{\id}{\text{id}}
\newcommand{\Id}{\text{Id}}

\newcommand{\op}{\ensuremath{\mathrm{op}}}

\newcommand{\unitaryeqq}{\ensuremath{\mathrel{\mathop *}=}}

\newcommand{\Ob}{\mathrm{Ob}}

\newcommand{\Set}{\ensuremath{\mathbf{Set}}}
\newcommand{\SET}{\ensuremath{\mathbf{Set}}}

\newcommand{\AAA}{\ensuremath{\mathbf{A}}}

\newcommand{\CC}{\ensuremath{\mathbf{C}}}
\newcommand{\VV}{\ensuremath{\mathbf{V}}}

\newcommand{\Wstar}{\ensuremath{\mathbf{W}^*_{\mathrm{NCPSU}}}}
\newcommand{\WNMIU}{\ensuremath{\mathbf{W}^*_{\mathrm{NMIU}}}}

\newcommand{\cpo}{\ensuremath{\mathbf{DCPO}}}

\newcommand{\cpobs}{\ensuremath{\mathbf{DCPO}_{\perp !}}}

\newcommand{\bit}{\textbf{bit}}
\newcommand{\qbit}{\textbf{qbit}}
\newcommand{\newunit}{\textbf{new unit}}
\newcommand{\newqbit}{\textbf{new qbit}}
\newcommand{\discard}{\textbf{discard}}
\newcommand{\sskip}{\textbf{skip}}

\newcommand{\while}[2]{\ensuremath{\textbf{while}\ $#1$\ \textbf{do}\ $#2$}}
\newcommand{\lleft}{\textbf{left}}
\newcommand{\rright}{\textbf{right}}
\newcommand{\ccase}{\textbf{case}}
\newcommand{\fold}{\textbf{fold}}
\newcommand{\ccopy}{\textbf{copy}}
\newcommand{\ZERO}{\mathbf{0}}

\newcommand{\sfold}{\ensuremath{\mathrm{fold}}}
\newcommand{\sunfold}{\ensuremath{\mathrm{unfold}}}
\newcommand{\ifold}{\ensuremath{\mathpzc{fold}}}
\newcommand{\iunfold}{\ensuremath{\mathpzc{unfold}}}

\newcommand{\fff}{\texttt{ff}}
\newcommand{\ttt}{\texttt{tt}}

\newcommand{\ffold}{\textbf{fold}}

\newcommand{\lrb}[1]{{\llbracket #1 \rrbracket}}
\newcommand{\flrb}[1]{{\llparenthesis #1 \rrparenthesis}}

\newcommand{\naturalto}{\ensuremath{\Rightarrow}}

\newcommand{\Halt}{\ensuremath{\mathrm{Halt}}}

\usetikzlibrary{decorations.pathmorphing}
\usetikzlibrary{decorations.markings}
\usetikzlibrary{decorations.pathreplacing}
\usetikzlibrary{arrows}
\usetikzlibrary{shapes}

\pgfdeclarelayer{edgelayer}
\pgfdeclarelayer{nodelayer}
\pgfsetlayers{edgelayer,nodelayer,main}

\tikzstyle{braceedge}=[decorate,decoration={brace,amplitude=10pt}]
\tikzstyle{square box}=[rectangle,fill=white,draw=black,minimum height=6mm,minimum width=6mm,yshift=0.7mm]
\tikzstyle{wire label}=[font=\footnotesize, auto,swap]

\tikzstyle{none}=[inner sep=0pt]
\tikzstyle{gn}=[circle,fill=Lime,draw=Black,line width=0.8 pt]
\tikzstyle{rn}=[circle,fill=Red,draw=Black, line width=0.8 pt]
\tikzstyle{H}=[rectangle,fill=Yellow,draw=Black]
\tikzstyle{line}=[scalar,fill=White,draw=Black]
\tikzstyle{io}=[rectangle,fill=White,draw=Black]
\tikzstyle{block}=[rectangle,fill=Orange,draw=Black]
\tikzstyle{graph}=[circle,fill=White,draw=Black]
\tikzstyle{empty}=[rectangle,fill=none,draw=none]
\tikzstyle{scaled}=[rectangle,fill=none,draw=none, font=\small]
\tikzstyle{box}=[rectangle,fill=White,draw=Black]
\tikzstyle{dot}=[circle,fill=Black,draw=Black,inner sep=0pt,minimum size=1pt]
\tikzstyle{small dot}=[circle,fill=Black,draw=Black,inner sep=0pt,minimum size=1pt]
\tikzstyle{Dot}=[circle,fill=Black,draw=Black,inner sep=0pt,minimum size=3pt]
\tikzstyle{diam}=[rectangle,fill=Black,draw,yscale=1.2,rotate=45]
\tikzstyle{gangle}=[rectangle,fill=Lime,draw=Black]
\tikzstyle{rangle}=[rectangle,fill=Red,draw=Black]
\tikzstyle{circ}=[circle,fill=none,draw=Black,scale=1.3]
\tikzstyle{ellip}=[ellipse,fill=none,draw=Black,scale=1.3,minimum width =1.3cm]
\tikzstyle{ellip2}=[ellipse,fill=White,draw=Black,scale=1.3,minimum width =3cm]
\tikzstyle{bbox}=[rectangle,fill=Blue,draw=Blue,scale=0.6]
\tikzstyle{gg}=[shape=rectangle,fill=White,draw=Black,dashed]

\tikzstyle{white circle}=[circle,fill=none,draw=Black,scale=1]
\tikzstyle{black circle}=[circle,fill=Black,draw=Black,scale=1]
\tikzstyle{grey circle}=[circle,fill=Gray,draw=Black,scale=1]
\tikzstyle{white rectangle}=[rectangle,fill=none,draw=Black,scale=1]

\tikzstyle{nodev}=[circle,fill=none,draw=Black,scale=1]
\tikzstyle{greynode}=[circle,fill=Grey,draw=Black,scale=1]
\tikzstyle{blacknode}=[circle,fill=Black,draw=Black,scale=1]
\tikzstyle{wirev}=[circle,fill=Black,draw=Black,inner sep=0pt,minimum size=3pt]
\tikzstyle{wirevred}=[circle,fill=Red,draw=Black,inner sep=0pt,minimum size=3pt]

\tikzstyle{simple}=[-,draw=Black]
\tikzstyle{to}=[->,draw=Black]
\tikzstyle{naturalto}=[-{Implies},double distance=1.5pt]
\tikzstyle{bdirected}=[<->,draw=Black]
\tikzstyle{bothdirs}=[bdirected,draw=Black]
\tikzstyle{bothdirsred}=[bdirected,draw=Red]
\tikzstyle{blue}=[-,draw=Blue]
\tikzstyle{redd}=[directed,draw=Red]
\tikzstyle{redu}=[-,draw=Red]
\tikzstyle{blued}=[directed,draw=Blue]
\tikzstyle{dash}=[dashed,draw=Black]
\tikzstyle{ddash}=[->,dashed,draw=Black]
\tikzstyle{dashedd}=[->,dashed]
\tikzstyle{dashedred}=[dashed,draw=Red]

\tikzstyle{equal-arrow}=[double equal sign distance]

\tikzstyle{dotpic}=[scale=0.5]

\tikzstyle{every picture}=[baseline=-0.25em]

\newcommand{\stikz}[2][1]{\scalebox{#1}{
\InputIfFileExists{#2}{}{\input{./tikz/#2}}
}}
\newcommand{\cstikz}[2][1]{\begin{center}\stikz[#1]{#2}\end{center}}

\begin{document}

\title{Quantum Programming with Inductive Datatypes: Causality and Affine Type Theory}
\titlerunning{Quantum Programming with Inductive Datatypes}

\author{Romain P\'echoux\inst{1} \and Simon Perdrix\inst{1} \and Mathys Rennela\inst{2} \and \\ Vladimir Zamdzhiev\inst{1}}
\authorrunning{R. P\'echoux\inst{1} \and S. Perdrix\inst{1} \and M. Rennela\inst{2} \and V. Zamdzhiev\inst{1}}

\institute{Universit\'e de Lorraine, CNRS, Inria, LORIA, France \and
Leiden Inst. Advanced Computer Sciences, Universiteit Leiden, The Netherlands
}

\maketitle

\begin{abstract}
Inductive datatypes in programming languages allow users to define useful data
structures such as natural numbers, lists, trees, and others. In this paper we
show how inductive datatypes may be added to the quantum programming language
QPL. We construct a sound categorical model for the language and by doing so we
provide the first detailed semantic treatment of user-defined inductive
datatypes in quantum programming.
We also show our denotational interpretation is invariant with respect to big-step
reduction, thereby establishing another novel result for quantum programming.
Compared to classical programming, this property is considerably more difficult to prove and we demonstrate its usefulness by showing how it immediately implies computational adequacy at all types.
To further cement our
results, our semantics is entirely based on a physically natural
model of von Neumann algebras, which are mathematical structures used by
physicists to study quantum mechanics.

\end{abstract}

\section{Introduction}

\emph{Quantum computing} is a computational paradigm which takes advantage of
quantum mechanical phenomena to perform computation. Much of the interest in
quantum computing is due to impressive algorithms,
such as Shor's polynomial-time algorithm for integer factorisation ~\cite{shor}, Grover's $O(\sqrt n)$ algorithm for unsorted database search~\cite{grover} and
others~\cite{quantum-algorithms-survey}.
Of course, in order to implement these algorithms it is necessary
to also develop appropriate \emph{quantum programming languages}. The present
paper is concerned with this issue and makes several theoretical contributions towards
the design and denotational semantics of quantum programming languages.



Our development is based around the quantum programming language QPL \cite{qpl}
which we extend with inductive datatypes.
Our paper is the first to construct a denotational semantics for user-defined inductive datatypes in quantum programming.
In the spirit of the original QPL,
our type system is \emph{affine} (discarding of arbitrary variables is
allowed, but copying is restricted). We also extend QPL with a copy operation
for \emph{classical data}, because this is an admissible operation in quantum
mechanics which improves programming convenience. The
addition of inductive datatypes requires a departure from the original
denotational semantics of QPL, which are based on finite-dimensional quantum
structures, and we consider instead (possibly infinite-dimensional) quantum
structures based on \emph{W*-algebras} (also known as \emph{von Neumann
algebras}), which have been used by physicists in the study of quantum
foundations~\cite{takesaki}.  As such, our semantic treatment is physically
natural and our model is more accessible to physicists and
experts in quantum computing compared to most other denotational models. 

QPL is a first-order programming language which has \emph{procedures}, but it
does not have lambda abstractions. Thus, there is no use for a !-modality
and we show how to model the copy operation by describing the canonical
comonoid structure of all classical types (including the inductive ones).

An important notion in quantum mechanics is the idea of \emph{causality} which
has been formulated in a variety of different ways.  In this paper, we consider
a simple operational interpretation of causality: if the output of a physical
process is discarded, then it does not matter which process
occurred~\cite{aleks-sander}.  In a symmetric monoidal category $\CC$ with
tensor unit $I$, this can be understood as requiring that for any morphism
(process) $f: A_1 \to A_2$, it must be the case that $\diamond_{A_2} \circ f =
\diamond_{A_1}$, where $\diamond_{A_i} : A_i \to I$ is the discarding map
(process) at the given objects.  This notion ties in very nicely with our affine
language, because we have to show that the interpretation of
values is causal, i.e., values are always discardable.

A major contribution of this paper is that we prove the denotational
semantics is invariant with respect to both small-step reduction and big-step
reduction. The latter is quite difficult in quantum programming
and our paper is the first to demonstrate such a result. Combining this with
the causal properties of values, we obtain (as a corollary) computational adequacy at all types.
\section{Syntax of QPL}\label{sec:syntax}

The syntax of QPL (including our extensions) is summarised in Figure~\ref{fig:syntax}.
A well-formed type context, denoted $\vdash \Theta$, is simply a list of distinct type variables.
A type $A$ is well-formed in type context $\Theta$, denoted $\Theta \vdash A,$ if the judgement can be derived according to the following rules:
\[
{\small{
    \begin{bprooftree}
    \AxiomC{{$\vdash \Theta$}}
    \UnaryInfC{$\Theta \vdash \Theta_i$}
    \end{bprooftree}
    \begin{bprooftree}
    \AxiomC{\phantom{$\Theta \vdash A$}}
    \UnaryInfC{$\Theta \vdash I$}
    \end{bprooftree}
    \begin{bprooftree}
    \AxiomC{\phantom{$\Theta \vdash A$}}
    \UnaryInfC{$\Theta \vdash \qbit$}
    \end{bprooftree}
    \begin{bprooftree}
    \AxiomC{$\Theta \vdash A$}
    \AxiomC{$\Theta \vdash B$}
    \RightLabel{$\star \in \{+, \otimes\}$}
    \BinaryInfC{$\Theta \vdash A \star B$}
    \end{bprooftree}
    \begin{bprooftree}
    \AxiomC{$\Theta, X \vdash A$}
    \UnaryInfC{$\Theta \vdash \mu X.A$}
    \end{bprooftree}
}}
\]
A type $A$ is \emph{closed} if $\cdot \vdash A$. Note that nested type induction is allowed.
Henceforth, we implicitly assume that all types we are dealing with are well-formed.

%
%
\begin{figure}
{\small{
\begin{tabular}{l  l  l  l}
  Type Variables           & $X,Y,Z$ & & \\
  Term Variables           & $x, y, q, b, u$ & & \\
  Procedure Names          & $f, g$ & & \\
	Types                    & $A, B, C$ & ::= &  $X$ | $I$ | \textbf{qbit} | $A+B$ | $A\otimes B$ | $\mu X. A$\\
	Classical Types          & $P, R$ & ::= &  $X$ | $I$ | $P+R$ | $P\otimes R$ | $\mu X. P$\\
  Terms                    & $M, N$ & ::= & \newunit\ $u$ | \discard\ $x$ | $y = \textbf{copy}\ x$ | \newqbit\ $q$ | \\
    & &                                   &  $b$ = \textbf{measure} $q$ | $q_1, \ldots, q_n\ \unitaryeqq S$  | $  M;N$ | \sskip\  |   \\
    & &                                   &  \while{$b$}{$M$} | $x = $ \textbf{left}$_{A,B} M$ | $x$ = \textbf{right}$_{A,B} M$ | \\ 
    & &                                   & \textbf{case} $y$ \textbf{of} $\{$\textbf{left} $x_1\to M\ |$ \textbf{right} $x_2 \to N\}$ | \\ 
    & &                                   & $x = (x_1, x_2) $ | $(x_1, x_2) = x$ | $y$ = \textbf{fold} $x$ | $y$ = \textbf{unfold} $x$ | \\
    & &                                   & \textbf{proc} $f ::\ x:A \to y:B\ \{M\}\ $ | $y = f(x)$\\
  Type contexts            & $\Theta$ &  ::= & $X_1, X_2, \ldots, X_n $\\
  Variable contexts        & $\Gamma, \Sigma$ &  ::= & $x_1: A_1, \ldots, x_n : A_n$\\
  Procedure contexts       & $\Pi$ &     ::= & $f_1: A_1 \to B_1, \ldots, f_n : A_n \to B_n$ \\
  Type Judgements          & \multicolumn{2}{l}{$\Theta \vdash A$} &\\
  Term Judgements          & \multicolumn{2}{l}{$\Pi \vdash \langle \Gamma \rangle\ M\ \langle \Sigma \rangle$} &
\end{tabular}
}}
\caption{Syntax of QPL.}
\label{fig:syntax}
\end{figure}
%
%
\begin{figure}
{
\small{
  \centerline{
    \begin{bprooftree}
    \AxiomC{}
    \RightLabel{(unit)} \UnaryInfC{$\Pi \vdash \langle \Gamma \rangle\ \newunit\ u\ \langle \Gamma, u:I \rangle$}
    \end{bprooftree}
    \begin{bprooftree}
    \AxiomC{}
    \RightLabel{(discard)} \UnaryInfC{$\Pi \vdash \langle \Gamma, x:A \rangle\ \textbf{discard}\ x\ \langle \Gamma \rangle$}
    \end{bprooftree}
  }
  \vspace{4mm}
  \centerline{
    \begin{bprooftree}
    \AxiomC{$P$ is a classical type}
    \RightLabel{(copy)} \UnaryInfC{$\Pi \vdash \langle \Gamma, x:P \rangle\ y = \textbf{copy}\ x\ \langle \Gamma, x : P, y: P \rangle$}
    \end{bprooftree}
    \begin{bprooftree}
    \AxiomC{\phantom{P}}
    \RightLabel{(skip)} \UnaryInfC{$\Pi \vdash \langle \Gamma \rangle\ \sskip\ \langle \Gamma \rangle$}
    \end{bprooftree}
  }
  \vspace{4mm}
  \centerline{
    \begin{bprooftree}
    \AxiomC{$\Pi \vdash \langle \Gamma \rangle\ M\ \langle \Gamma' \rangle$}
    \AxiomC{$\Pi \vdash \langle \Gamma' \rangle\ N\ \langle \Sigma \rangle$}
    \RightLabel{(seq)} \BinaryInfC{$\Pi \vdash \langle \Gamma \rangle\ M;N\ \langle \Sigma \rangle$}
    \end{bprooftree}
    \begin{bprooftree}
    \AxiomC{$\Pi \vdash \langle \Gamma, b: \textbf{bit} \rangle\ M\ \langle  \Gamma, b: \textbf{bit} \rangle$}
    \RightLabel{(while)} \UnaryInfC{$\Pi \vdash \langle \Gamma, b: \textbf{bit} \rangle\ 
      \while{$b$}{$M$}\ \langle \Gamma, b: \textbf{bit} \rangle$}
    \end{bprooftree}
   }
  \vspace{4mm}
  \centerline{
    \begin{bprooftree}
    \AxiomC{}
    \RightLabel{(qbit)} \UnaryInfC{$\Pi \vdash \langle \Gamma \rangle\ \newqbit\ q\ \langle \Gamma, q:\textbf{qbit} \rangle$}
    \end{bprooftree}
    \begin{bprooftree}
    \AxiomC{}
    \RightLabel{(measure)} \UnaryInfC{$\Pi \vdash \langle \Gamma, q: \qbit \rangle\ 
      b = \textbf{measure}\ q\ \langle \Gamma, b: \textbf{bit} \rangle$}
    \end{bprooftree}
   }
  \vspace{4mm}
  \centerline{
    \begin{bprooftree}
    \AxiomC{$S$ is a unitary of arity $n$}
    \RightLabel{(unitary)} \UnaryInfC{$\Pi \vdash \langle \Gamma, q_1 :\qbit, \ldots, q_n : \qbit \rangle\
    q_1, \ldots, q_n \unitaryeqq S\ \langle \Gamma, q_1 :\qbit, \ldots, q_n : \qbit \rangle $}
    \end{bprooftree}
   }
  \vspace{4mm}
   \centerline{
    \begin{bprooftree}
    \AxiomC{}
    \RightLabel{(left)} \UnaryInfC{$\Pi \vdash \langle \Gamma, x:A \rangle\ y = \lleft_{A,B}\ x\ \langle \Gamma, y: A+B \rangle$}
    \end{bprooftree}
    \begin{bprooftree}
    \AxiomC{}
    \RightLabel{(right)} \UnaryInfC{$\Pi \vdash \langle \Gamma, x:B \rangle\ y = \rright_{A,B}\ x\ \langle \Gamma, y: A+B \rangle$}
    \end{bprooftree}
  }
  \vspace{4mm}
  \centerline{
    \begin{bprooftree}
    \AxiomC{$\Pi \vdash \langle \Gamma, x_1: A \rangle\ M_1\ \langle \Sigma \rangle$}
    \AxiomC{$\Pi \vdash \langle \Gamma, x_2: B \rangle\ M_2\ \langle \Sigma \rangle$}
    \RightLabel{(case)} \BinaryInfC{$\Pi \vdash \langle \Gamma, y: A+B \rangle\ 
      \ccase\ y\ \textbf{of}\ \{\lleft_{A,B}\ x_1 \to M_1\ |\ \rright_{A,B}\ x_2 \to M_2\ \}\ \langle \Sigma \rangle$}
    \end{bprooftree}
    \qquad\ \ 
  }
  \vspace{4mm}
  \centerline{
    \begin{bprooftree}
    \AxiomC{}
    \RightLabel{(pair)} \UnaryInfC{$\Pi \vdash \langle \Gamma, x_1: A, x_2 : B \rangle\ x = (x_1, x_2)\ \langle \Gamma, x: A \otimes B \rangle$}
    \end{bprooftree}
    \begin{bprooftree}
    \AxiomC{}
    \RightLabel{(unpair)} \UnaryInfC{$\Pi \vdash \langle \Gamma, x: A \otimes B \rangle\ (x_1, x_2) = x \ \langle \Gamma, x_1: A, x_2 : B \rangle$}
    \end{bprooftree}
  }
  \vspace{4mm}
  \centerline{
    \begin{bprooftree}
    \AxiomC{}
    \RightLabel{(fold)} \UnaryInfC{$\Pi \vdash \langle \Gamma, x: A[\mu X. A / X] \rangle\ y = \textbf{fold}_{\mu X. A}\ x\ \langle \Gamma, y: \mu X.A \rangle$}
    \end{bprooftree}
    \begin{bprooftree}
    \AxiomC{}
    \RightLabel{(unfold)} \UnaryInfC{$\Pi \vdash \langle \Gamma, x: \mu X. A \rangle\ y = \textbf{unfold}\ x\ \langle \Gamma, y: A[\mu X. A / X] \rangle$}
    \end{bprooftree}
  }
  \vspace{4mm}
  \centerline{
    \begin{bprooftree}
    \AxiomC{$\Pi, f: A \to B \vdash \langle x:A \rangle\ M\ \langle y:B \rangle$}
    \RightLabel{(proc)} \UnaryInfC{$\Pi \vdash \langle \Gamma \rangle\ \textbf{proc}\ f ::\ x:A \to y:B\ \{M\}\ \langle \Gamma \rangle$}
    \end{bprooftree}
%
    \begin{bprooftree}
    \AxiomC{$\phantom{\langle}$}
    \RightLabel{(call)} \UnaryInfC{$\Pi, f: A \to B \vdash \langle \Gamma, x: A \rangle\ y = f(x)\ \langle \Gamma, y:B \rangle$}
    \end{bprooftree}
  }
  }
}
\caption{Formation rules for QPL terms.}
\label{fig:term-formation}
\end{figure}

\begin{example}
\label{ex:syntax}
The type of natural numbers is defined as $\textbf{Nat} \equiv \mu X. I + X.$
Lists of a closed type $\cdot \vdash A$ are defined as $\textbf{List}(A) \equiv \mu Y. I + A \otimes Y.$
\end{example}

Notice that our type system is not equipped with a !-modality. Indeed,
in the absence of function types, there is no reason to introduce it.
Instead, we specify the subset of types where copying is an admissible
operation. The \emph{classical types} are a subset of our types defined in
Figure~\ref{fig:syntax}.
They are characterised by the property that variables
of classical types may be copied, whereas variables of non-classical types may
not be copied (see (copy) rule in Figure~\ref{fig:term-formation}). 

%

We use small latin letters (e.g. $x,y,u,q,b$) to range over \emph{term variables}.
More specifically, $q$ ranges over variables of type $\qbit$, $u$ over
variables of unit type $I$, $b$ over variables of type $\textbf{bit} \coloneqq
I+I$ and $x,y$ range over variables of arbitrary type. We use $\Gamma$ and
$\Sigma$ to range over \emph{variable contexts}. A variable context is a
function from term variables to \emph{closed types}, which we write as $\Gamma
= x_1: A_1, \ldots, x_n:A_n.$

We use $f,g$ to range over \emph{procedure names}.  Every procedure name $f$ has
an \emph{input type} $A$ and an \emph{output type} $B$, denoted $f: A\to B$,
where $A$ and $B$ are closed types. We use $\Pi$ to range over \emph{procedure
contexts}. A procedure context is a function from procedure names to pairs of
procedure input-output types, denoted
$\Pi = f_1: A_1 \to B_1, \ldots , f_n : A_n \to B_n.$


\begin{remark}
Unlike lambda abstractions, procedures cannot be passed to other procedures as
input arguments, nor can they be returned as output.
\end{remark}

A \emph{term judgement} has the form $\Pi \vdash \langle \Gamma \rangle\ M\ \langle \Sigma \rangle$ (see Figure~\ref{fig:term-formation})
and indicates that term $M$ is well-formed in procedure context $\Pi$ with
input variable context $\Gamma$ and output variable context $\Sigma$. All types occurring
within it are closed. 

The intended interpretation of the quantum rules are as follows.
The (qbit) rule prepares a new qubit $q$ in state $\ket 0 \bra 0$.
The (unitary) rule applies a unitary operator $S$ to a sequence of qubits in the standard way.
The (measure) rule
performs a quantum measurement on qubit $q$ and stores the measurement outcome
in bit $b$. The measured qubit is destroyed in the process.

The no-cloning theorem of quantum mechanics~\cite{no-cloning} shows that arbitrary qubits cannot be copied. Because of this,
copying is restricted only to classical types, as indicated by the (copy) rule, and this allows us to avoid runtime errors.
Like the original QPL~\cite{qpl}, our type system is also
\emph{affine} and the (discard) rule allows any variable to be discarded.

\section{Operational Semantics of QPL}\label{sec:operational}

In this section we describe the operational semantics of QPL. The central
notion is that of a \emph{program configuration} which provides a complete
description of the current state of program execution. It
consists of four components that must satisfy some coherence properties:
(1) the term which remains to be executed; (2) a \emph{value assignment}, which
is a function that assigns formal expressions to variables as a result of
execution; (3) a \emph{procedure store} which keeps track of what procedures
have been defined so far and (4) the \emph{quantum state} computed
so far.

\paragraph{Value Assignments.}
A \emph{value} is an expression defined by the following grammar:
\[v, w ::= *\ |\ n\ |\ \lleft_{A,B} v\ |\ \rright_{A,B} v\ |\ (v,w)\ |\ \ffold_{\mu X.A} v \]
where $n$ ranges over the natural numbers.
Think of $*$ as representing the unique value of unit type $I$
and of $n$ as representing a pointer to the $n$-th qubit of a quantum state $\rho$.
Specific values of interest are $\texttt{ff}:=
\textbf{left}_{I,I} *$ and $\texttt{tt}:= \textbf{right}_{I,I}
*$ which correspond to \textbf{false} and \textbf{true} respectively.

A \emph{qubit pointer context} is a set $Q$ of natural numbers.
A value $v$ of type $A$ is well-formed in qubit pointer context $Q$,
denoted $Q \vdash v : A$, if the judgement is derivable from the following rules:
  {\small{
  \[
    \begin{bprooftree}
    \AxiomC{\phantom{$Q$}}
    \RightLabel{} \UnaryInfC{$\cdot \vdash *: I$}
    \end{bprooftree}
    \begin{bprooftree}
    \AxiomC{\phantom{$Q$}}
    \UnaryInfC{$\{n\} \vdash n : \qbit$}
    \end{bprooftree}
    \begin{bprooftree}
    \AxiomC{$Q \vdash v : A$}
    \UnaryInfC{$Q \vdash \lleft_{A,B} v : A+B$}
    \end{bprooftree}
    \begin{bprooftree}
    \AxiomC{$Q \vdash v : B$}
    \UnaryInfC{$Q \vdash \rright_{A,B} v : A+B$}
    \end{bprooftree}
  \]
  \[
    \begin{bprooftree}
    \AxiomC{$Q_1 \vdash v : A$}
    \AxiomC{$Q_2 \vdash w : B$}
    \AxiomC{$Q_1 \cap Q_2 = \varnothing$}
    \TrinaryInfC{$Q_1, Q_2 \vdash (v, w) : A \otimes B$}
    \end{bprooftree}
    \begin{bprooftree}
    \AxiomC{$Q \vdash v : A[\mu X. A / X]$}
    \UnaryInfC{$Q \vdash \fold_{\mu X.A} v : \mu X. A$}
    \end{bprooftree}
  \]}}
If $v$ is well-formed, then its type and qubit pointer context are uniquely
determined. If $Q \vdash v : P$ with $P$ classical, then we say $v$ is a \emph{classical value}.

\begin{lemma}
\label{lem:classical-value-syntax}
If $Q \vdash v : P$ is a well-formed classical value, then $Q = \cdot.$
\end{lemma}

A \emph{value assignment} is a function from term variables to values, which we write as
$V =\{ x_1 = v_1, \ldots, x_n = v_n \},$
where $x_i$ are variables and $v_i$ are values. A
value assignment is \emph{well-formed} in qubit pointer context $Q$ and variable context $\Gamma$, denoted $Q ; \Gamma \vdash V,$ if $V$ has exactly the same variables as $\Gamma$
and $Q=Q_1,\ldots, Q_n$, s.t. $Q_i \vdash v_i : A_i.$
Such a splitting of $Q$ is necessarily unique, if it exists, and some of the $Q_i$ may be empty.

\paragraph{Procedure Stores.}
A \emph{procedure store} is a set of procedure definitions, written as:
\[
\Omega = \left\{
          f_1 :: x_1 : A_1 \to y_1 : B_1\ \{ M_1 \}, \ldots ,
          f_n :: x_n : A_n \to y_n : B_n\ \{ M_n \}
          \right\}.
\]
A procedure store is \emph{well-formed} in procedure context $\Pi$, written $\Pi \vdash \Omega$, if the judgement is derivable via the following rules: 
  \[
    \begin{bprooftree}
    \AxiomC{\phantom{$\Pi$}}
    \UnaryInfC{$\cdot \vdash \cdot$}
    \end{bprooftree}
    \qquad
    \begin{bprooftree}
    \AxiomC{$\Pi \vdash \Omega$}
    \AxiomC{$\Pi, f : A \to B \vdash \langle x:A \rangle\ M\ \langle y:B \rangle$}
    \BinaryInfC{$\Pi, f : A \to B \vdash  \Omega, f :: x:A \to y:B\ \{M\}$}
    \end{bprooftree}
  \]
\paragraph{Program configurations.}
A \emph{program configuration} is a quadruple $(M\ |\ V\ |\ \Omega\ |\ \rho),$ where $M$ is a term, $V$ is a value assignment, $\Omega$ is a procedure store and $\rho \in \mathbb{C}^{2^n \times 2^n}$ is
a finite-dimensional density matrix with $0 \leq \trace \rho \leq 1.$ The density matrix $\rho$ represents
a (mixed) quantum state and its trace may be smaller than one because we also use it to encode probability information (see Remark~\ref{rem:nondeterminism}).
We write $\emph{dim}(\rho) = n$ to indicate that the dimension of $\rho$ is $n.$ 

A \emph{well-formed} program configuration is a configuration $(M\ |\ V\ |\ \Omega\ |\ \rho),$ where there exist (necessarily unique) $\Pi, \Gamma, \Sigma, Q$, such that:
(1) $\Pi \vdash \langle \Gamma \rangle\ M\ \langle \Sigma \rangle$ is a well-formed term;
(2) $ Q ; \Gamma \vdash V$ is a well-formed value assignment;
(3) $\Pi \vdash \Omega$ is a well-formed procedure store;
and (4) $Q = \{1,2, \ldots, \emph{dim}(\rho)\}.$
We write $\Pi; \Gamma; \Sigma; Q \vdash (M\ |\ V\ |\ \Omega\ |\ \rho)$ to indicate this situation.
The formation rules enforce that the qubits of $\rho$ and the qubit pointers from $V$ are in a 1-1 correspondence.

\begin{figure}
{\small{
  \centerline{
    \begin{bprooftree}
    \AxiomC{}
    \RightLabel{(unit)} \UnaryInfC{$
      (\newunit\ u\ |\ V\ |\ \Omega\ |\ \rho)
      \leadsto
      (\sskip\ |\ V, u=*\ |\ \Omega\ |\ \rho)
    $}
    \end{bprooftree}
  }
  \vspace{0.5mm}
  \vspace{4mm}\centerline{
    \begin{bprooftree}
    \AxiomC{}
    \RightLabel{(discard)} \UnaryInfC{$
      (\textbf{discard}\ x\ |\ V, x=v\ |\ \Omega\ |\ \rho)
      \leadsto
      (\sskip\ |\ r_v(V)\ |\ \Omega\ |\ tr_v(\rho))
    $}
    \end{bprooftree}
  }
  \vspace{0.5mm}
  \vspace{4mm}\centerline{
    \begin{bprooftree}
    \AxiomC{}
    \RightLabel{(copy)} \UnaryInfC{$
      (y = \ccopy\ x\ |\ V, x=v\ |\ \Omega\ |\ \rho)
      \leadsto
      (\sskip\ |\ V, x=v, y=v \ |\ \Omega\ |\ \rho)
    $}
    \end{bprooftree}
  }
  \vspace{0.5mm}
  \vspace{4mm}\centerline{
    \begin{bprooftree}
    \AxiomC{}
    \RightLabel{(qbit)} \UnaryInfC{$
    (\newqbit\ q\ |\ V\ |\ \Omega\ |\ \rho)
      \leadsto
    (\sskip\ |\ V, q = \emph{dim}(\rho)+1 \ |\ \Omega\ |\ \rho \otimes \ket 0 \bra 0)
    $}
    \end{bprooftree}
  }
  \vspace{0.5mm}
  \vspace{4mm}\centerline{
    \begin{bprooftree}
    \AxiomC{}
    \RightLabel{(unitary)} \UnaryInfC{
      $(\vec q \unitaryeqq S\ |\ V, \vec q = \vec m \ |\ \Omega\ |\ \rho)
      \leadsto
      (\sskip |\ V, \vec q = \vec m\ |\ \Omega\ |\ S_{\vec m}(\rho))
    $}
    \end{bprooftree}
  }
  \vspace{0.5mm}
  \vspace{4mm}\centerline{
    \begin{bprooftree}
    \AxiomC{}
    \RightLabel{(measure0)} \UnaryInfC{$
      (b = \textbf{measure}\ q\ |\ V, q = m\ |\ \Omega\ |\ \rho)
      \leadsto
      (\sskip\ |\ r_m(V), b = \fff\ |\ \Omega\ |\ \mbox{}_m \bra 0 \rho \ket 0_m)
    $}
    \end{bprooftree}
  }
  \vspace{0.5mm}
  \vspace{4mm}\centerline{
    \begin{bprooftree}
    \AxiomC{}
    \RightLabel{(measure1)} \UnaryInfC{$
      (b = \textbf{measure}\ q\ |\ V, q = m \ |\ \Omega\ |\ \rho)
      \leadsto
      (\sskip\ |\ r_m(V), b = \ttt \ |\ \Omega\ |\ \mbox{}_m \bra 1 \rho \ket 1_m)
    $}
    \end{bprooftree}
  }
  \vspace{0.5mm}
  \vspace{4mm}\centerline{
    \begin{bprooftree}
    \AxiomC{}
    \RightLabel{(seq1)} \UnaryInfC{$(\sskip;P\ |\ V\ |\ \Omega\ |\ \rho) \leadsto (P\ |\ V\ |\ \Omega\ |\ \rho)$}
    \end{bprooftree}
    \qquad
  }
  \vspace{0.5mm}
  \vspace{4mm}\centerline{
    \begin{bprooftree}
    \AxiomC{$(P\ |\ V\ |\ \Omega\ |\ \rho) \leadsto (P'\ |\ V'\ |\ \Omega'\ |\ \rho')$}
    \RightLabel{(seq2)} \UnaryInfC{$(P;Q\ |\ V\ |\ \Omega\ |\ \rho) \leadsto (P';Q\ |\ V'\ |\ \Omega'\ |\ \rho')$}
    \end{bprooftree}
    \qquad
  }
  \vspace{0.5mm}
  \vspace{4mm}\centerline{
    \begin{bprooftree}
    \AxiomC{}
    \RightLabel{(while)} \UnaryInfC{$ 
      (\while{$b$}{$M$}\ |\ V, b = v\ |\ \Omega\ |\ \rho)
      \leadsto
      (\textbf{if}\ b\ \textbf{then}\ \{M;\while{$b$}{$M$}\}
 \ |\ V, b = v \ |\ \Omega\ |\ \rho)
    $}
    \end{bprooftree}
  }
  \vspace{0.5mm}
  \vspace{4mm}\centerline{
    \begin{bprooftree}
    \AxiomC{}
    \RightLabel{(left)} \UnaryInfC{$
      (y = \lleft\ x\ |\ V, x = v \ |\ \Omega\ |\ \rho)
      \leadsto
      (\sskip\ |\ V, y =\lleft\ v\ |\ \Omega\ |\ \rho)
    $}
    \end{bprooftree}
  }
  \vspace{0.5mm}
  \vspace{4mm}\centerline{
    \begin{bprooftree}
    \AxiomC{}
    \RightLabel{(right)} \UnaryInfC{$
      (y = \rright\ x\ |\ V, x = v \ |\ \Omega\ |\ \rho)
      \leadsto
      (\sskip\ |\ V, y =\rright\ v \ |\ \Omega\ |\ \rho)
    $}
    \end{bprooftree}
  }
  \vspace{0.5mm}
  \vspace{4mm}\centerline{
    \begin{bprooftree}
    \AxiomC{}
    \RightLabel{(caseLeft)} \UnaryInfC{$
      (\ccase\ y\ \textbf{of}\ \{\lleft\ x_1 \to M_1\ |\ \rright\ x_2 \to M_2\ \}\ |\ V, y=\lleft\ v\ |\ \Omega\ |\ \rho)
      \leadsto
      (M_1\ |\ V, x_1=v\ |\ \Omega\ |\ \rho)
    $}
    \end{bprooftree}
    \qquad\ \ 
  }
  \vspace{0.5mm}
  \vspace{4mm}\centerline{
    \begin{bprooftree}
    \AxiomC{}
    \RightLabel{(caseRight)} \UnaryInfC{$
      (\ccase\ y\ \textbf{of}\ \{\lleft\ x_1 \to M_1\ |\ \rright\ x_2 \to M_2\ \}\ |\ V, y=\rright\ v\ |\ \Omega\ |\ \rho)
      \leadsto
      (M_2\ |\ V, x_2=v\ |\ \Omega\ |\ \rho)
    $}
    \end{bprooftree}
    \qquad\ \ 
  }
  \vspace{0.5mm}
  \vspace{4mm}\centerline{
    \begin{bprooftree}
    \AxiomC{}
    \RightLabel{(pair)} \UnaryInfC{$
    (x = (x_1, x_2)\ |\ V, x_1 = v_1, x_2 = v_2 \ |\ \Omega\ |\ \rho)
    \leadsto
    (\sskip\ |\ V, x=(v_1,v_2)\ |\ \Omega\ |\ \rho)
    $}
    \end{bprooftree}
  }
  \vspace{0.5mm}
  \vspace{4mm}\centerline{
    \begin{bprooftree}
    \AxiomC{}
    \RightLabel{(unpair)} \UnaryInfC{$
    ((x_1, x_2) = x\ |\ V, x=(v_1,v_2)\ |\ \Omega\ |\ \rho)
    \leadsto
    (\sskip\ |\ V, x_1 = v_1, x_2 = v_2\ |\ \Omega\ |\ \rho)
    $}
    \end{bprooftree}
  }
  \vspace{0.5mm}
  \vspace{4mm}\centerline{
    \begin{bprooftree}
    \AxiomC{}
    \RightLabel{(fold)} \UnaryInfC{$
      (y = \textbf{fold}\ x\ |\ V, x = v\ |\ \Omega\ |\ \rho)
      \leadsto
      (\sskip\ |\ V, y =\textbf{fold}\ v\ |\ \Omega\ |\ \rho)
    $}
    \end{bprooftree}
  }
  \vspace{0.5mm}
  \vspace{4mm}\centerline{
    \begin{bprooftree}
    \AxiomC{}
    \RightLabel{(unfold)} \UnaryInfC{$
      (y = \textbf{unfold}\ x\ |\ V, x = \textbf{fold}\ v\ |\ \Omega\ |\ \rho)
      \leadsto
      (\sskip\ |\ V, y = v\ |\ \Omega\ |\ \rho)
    $}
    \end{bprooftree}
  }
  \vspace{0.5mm}
  \vspace{4mm}\centerline{
    \begin{bprooftree}
    \AxiomC{}
    \RightLabel{(proc)} \UnaryInfC{$
      (\textbf{proc}\ f::\ x:A \to y:B\ \{M\}\ |\ V\ |\ \Omega\ |\ \rho)
      \leadsto
      (\sskip\ |\ V\ |\ \Omega, f::\ x:A \to y:B\ \{M\}\ |\ \rho)
    $}
    \end{bprooftree}
  }
  \vspace{1mm}
  \vspace{4mm}\centerline{
    \begin{bprooftree}
    \AxiomC{}
    \RightLabel{(call)} \UnaryInfC{$
      (y_1 = f(x_1)\ |\ V, x_1 = v\ |\ \Omega, f::\ x_2:A \to y_2:B\ \{M\}\ |\ \rho)
      \leadsto
      (M_\alpha\ |\ V, x_1 =v\ |\ \Omega, f::\ x_2 : A \to y_2 : B\ \{M\}\ |\ \rho)
    $}
    \end{bprooftree}
  }
}}
\caption{Small Step Operational semantics of QPL.}
\label{fig:operational-small-step}
\end{figure}


The small step semantics is defined for configurations
$(M\ |\ V\ |\ \Omega\ |\ \rho)$ by induction on $M$ in Figure~\ref{fig:operational-small-step} 
and we now explain the notations used therein.

In the (discard) rule, we use two functions that depend on a value $v$. They are $tr_v,$ which modifies the quantum state $\rho$ by
tracing out all of its qubits which are used in $v$, and $r_v$ which simply reindexes the value assignment, so that the pointers
within $r_v(V)$ correctly point to the corresponding qubits of $tr_v(\rho),$ which is potentially of smaller dimension than $\rho$.
Formally, for a well-formed value $v$, let $Q$ and $A$ be the unique qubit pointer context and type, such that $Q \vdash v : A$. Then $tr_v(\rho)$ is the quantum
state obtained from $\rho$ by tracing out all qubits specified by $Q$.
Given a value assignment $V=\{x_1 = v_1, \ldots x_n = v_n\}$, then
$ r_v(V) = \{x_1 = r_v'(v_1), \ldots, x_n = r_v'(v_n)\}, \text{ where:} $
\[
r_v'(w)=
\begin{cases}
*,                                         & \text{if } w=* \\
k - |\{ i \in Q\ |\ i < k  \}|,         & \text{if } w=k \in \mathbb N \\
\textbf{left}\ r_v'(w'),                   & \text{if } w=\textbf{left}\ w'\\
\textbf{right}\ r_v'(w'),                  & \text{if } w=\textbf{right}\ w'\\
(r_v'(w_1),r_v'(w_2))                      & \text{if } w=(w_1,w_2) \\
\textbf{fold}\ r_v'(w'),                   & \text{if } w=\textbf{fold}\ w'
\end{cases}
\]



In the (unitary) rule, the superoperator $S_{\vec m}$ applies the unitary $S$ to the vector of qubits specified by $\vec m$.
In the (measure0) and (measure1) rules, the $m$-th qubit of $\rho$ is measured in the computational basis, the measured qubit is destroyed in the process and the measurement outcome is stored in the bit $b$.
More specifically, $\ket i_m = I_{2^{m-1}} \otimes \ket i \otimes I_{2^{n-m}}$
and $ _m \bra i$ is its adjoint, for $i \in \{0,1\}$, and where $I_{n}$ is the identity matrix in $\mathbb{C}^{n \times n}$.

\begin{remark}
\label{rem:nondeterminism}
Because of the way we decided to handle measurements, reduction $(-
\leadsto -)$ is a \emph{nondeterministic} operation, where we encode the
probabilities of reduction within the trace of our density matrices in a similar way
to~\cite{quantum-goi}. Equivalently, we may see the reduction relation as
\emph{probabilistic} provided that we normalise all density matrices and
decorate the reductions with the appropriate probability information as
specified by the Born rule of quantum mechanics. The nondeterministic view
leads to a more concise and clear presentation and because of this we have
chosen it over the probabilistic view.
\end{remark}

In the (while) rule, the term $\textbf{if}\ b\ \textbf{then}\ \{M\} $ is syntactic sugar for $\ccase\ b\ \textbf{of}\ \{\lleft\ u \to b = \lleft\ u\ |\ \rright\ u \to b = \rright\ u;M\ \}$. 
The (proc) rule simply defines a procedure which is added to the procedure store and the (call) rule is used to call already defined procedures.
In the (call) rule, the term $M_\alpha$ is $\alpha$-equivalent
to $M$ and is obtained from it by renaming the input $x_2$ to $x_1$, renaming the output $y_2$ to $y_1$ and renaming all other variables within $M$ to some fresh names, so as to avoid conflicts with the input, output and the rest of the variables within $V$.



\begin{theorem}[Subject reduction]\label{thm:subject}
If $\Pi; \Gamma; \Sigma; Q \vdash (M\ |\ V\ |\ \Omega\ |\ \rho)$ and $(M\ |\ V\ |\ \Omega\ |\ \rho) \leadsto (M'\ |\ V'\ |\ \Omega'\ |\ \rho')$,
then $\Pi'; \Gamma'; \Sigma; Q' \vdash (M'\ |\ V'\ |\ \Omega'\ |\ \rho')$, for some (necessarily unique) contexts $\Pi', \Gamma', Q'$ and where $\Sigma$ is invariant.
\end{theorem}

\begin{assumption}
From now on we assume all configurations are well-formed.
\end{assumption}

A configuration $(M\ |\ V\ |\ \Omega\ |\ \rho)$ is said to be \emph{terminal}
if $M = \sskip.$ Program execution finishes at terminal
configurations, which are characterised by the property that they do not
reduce any further. We will use calligraphic letters $(\mathcal C, \mathcal
D, \ldots)$ to range over configurations and we will use $\mathcal T$ to range over terminal configurations.
For a configuration $\mathcal C = (M\ |\ V\ |\ \Omega\ |\ \rho)$, we write for brevity $\trace{\mathcal C} \coloneqq \trace \rho $
and we shall say $\mathcal C$ is \emph{normalised} whenever $\trace{\mathcal C} = 1$.
We say that a configuration $\mathcal C$ is \emph{impossible} if $\trace{\mathcal C} = 0$ and we say it is \emph{possible} otherwise.

\begin{theorem}[Progress]\label{thm:progress}
If $\mathcal C$ is a configuration, then either $\mathcal C$ is terminal or there exists a configuration $\mathcal D$, such that 
$\mathcal C \leadsto \mathcal D.$ Moreover, if $\mathcal C$ is not terminal, then $\trace{\mathcal C} = \sum_{\mathcal C \leadsto \mathcal D} \trace{\mathcal D}$ and there are at most two such configurations $\mathcal D$.
\end{theorem}

\begin{figure}[t]
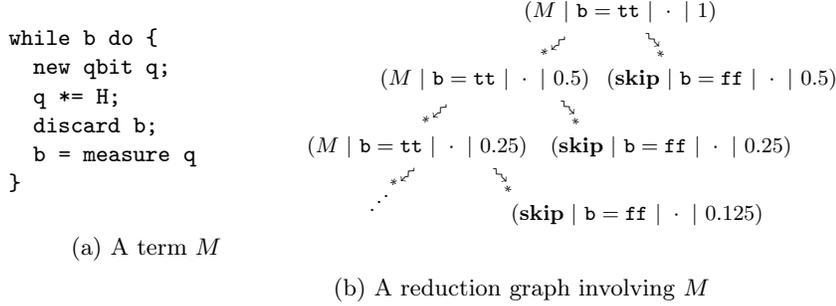

\begin{minipage}{0.3\textwidth}
\begin{subfigure}[b]{\textwidth}
\begin{verbatim}
while b do {
  new qbit q;
  q *= H;
  discard b;
  b = measure q
}
\end{verbatim}
\caption{A term $M$}
\label{sfig:term}
\end{subfigure}
\end{minipage}
\begin{minipage}{0.5\textwidth}
\begin{subfigure}[b]{\textwidth}
\cstikz[0.9]{execution.tikz}
\caption{A reduction graph involving $M$}
\label{sfig:rg}
\end{subfigure}
\end{minipage}
\caption{Example of a term and of a reduction graph.}
\label{fig:cointoss}
\end{figure}

In the above situation, the probability of reduction is $\mathrm{Pr}(\mathcal C \leadsto \mathcal D) \coloneqq \trace{\mathcal D} / \trace{\mathcal C},$
for any possible $\mathcal C$ (see Remark~\ref{rem:nondeterminism}) and Theorem~\ref{thm:progress} shows the total
probability of all single-step reductions is 1.
If $\mathcal C$ is impossible, then $\mathcal C$ occurs with probability 0 and subsequent reductions are also impossible.

\paragraph{Probability of Termination}
Given configurations $\mathcal C$ and $\mathcal D$ let $\Seq_n(\mathcal C, \mathcal D) \coloneqq \{ \mathcal C_0 \leadsto \cdots \leadsto \mathcal C_n |\ \mathcal C_0 = \mathcal C\ \text{and}\ \mathcal C_n = \mathcal D \}$,
and let $\Seq_{\leq n}(\mathcal C, \mathcal D) = \bigcup_{i=0}^n \Seq_n(\mathcal C, \mathcal D)$.
Finally, let
$ \TerSeq_{\leq n}(\mathcal C) \coloneqq \bigcup_{\mathcal T \text{ terminal}} \Seq_{\leq n}(\mathcal C, \mathcal T). $
In other words, $\TerSeq_{\leq n}(\mathcal C)$ is the set of all reduction sequences from $\mathcal C$ which terminate in at most $n$ steps (including 0 if $\mathcal C$ is terminal).
For every terminating reduction sequence $r = \left( \mathcal C \leadsto \cdots \leadsto \mathcal T \right)$,
let $\End(r) \coloneqq \mathcal T,$ i.e. $\End(r)$ is simply the (terminal) endpoint of the sequence.

For any configuration $\mathcal C$, the sequence $\left( \sum_{r \in \TerSeq_{\leq n}(\mathcal C)} \trace{\End(r)}\right)_{n \in \mathbb N}$
is increasing with upper bound $\trace{\mathcal C}$ (follows from Theorem~\ref{thm:progress}).
For any possible $\mathcal C$, we define:
\[ \Halt(\mathcal C) \coloneqq \bigvee_{n=0}^\infty \sum_{r \in \TerSeq_{\leq n}(\mathcal C)} \trace{\End(r)} / \trace{\mathcal C} \]
which is exactly the \emph{probability of termination} of $\mathcal C$.
This is justified, because $\Halt(\mathcal T) = 1$, for any terminal (and possible) configuration $\mathcal T$
and $\Halt(\mathcal C) = \sum_{\substack{\mathcal C \leadsto \mathcal D \\ \mathcal D\ \mathrm{possible}}} \mathrm{Pr}(\mathcal C \leadsto \mathcal D) \Halt(\mathcal D).$
We write $\leadsto_*$ for the transitive closure of $\leadsto$.

\begin{example}
Consider the term $M$ in Figure~\ref{fig:cointoss}. The body of the \textbf{while} loop (\ref{sfig:term}) has the effect of performing a fair coin toss (realised through quantum measurement in the standard way) and storing the outcome in variable $\texttt{b}$.
Therefore, starting from configuration $\mathcal C = (M\ |\ \texttt{b} = \texttt{tt}\ |\ \cdot\ |\ 1) $, as in Subfigure~\ref{sfig:rg}, the program has the effect of tossing a fair coin until \texttt{ff} shows up. The set of terminal configurations reachable from $\mathcal C$
is $\{ (\sskip\ |\ \texttt{b} = \texttt{ff}\ |\ \cdot\ |\ 2^{-i}) \ |\ i \in \mathbb N_{\geq 1}\}$ and the last component of each configuration is a $1 \times 1$ density matrix which is exactly the probability of reducing to the configuration.
Therefore $\Halt(\mathcal C) = \sum_{i = 1}^\infty 2^{-i} = 1$.
\end{example}

\begin{figure}[t]
\begin{minipage}{0.5\textwidth}
{\small{
\begin{subfigure}[b]{\textwidth}
\begin{verbatim}
proc GHZnext :: l : ListQ -> l : ListQ {
  new qbit q;
  case l of
      nil -> q*=H;
             l = q :: nil
    | q' :: l' -> q',q *= CNOT;
                  l = q :: q' :: l'
}

proc GHZ :: n : Nat -> l : ListQ {
  case n of
      zero -> l = nil
    | s(n') -> l = GHZnext(GHZ(n'))
}
\end{verbatim}
\caption{Procedures for generating GHZ$_n$}
\label{sfig:ghz}
\end{subfigure}
}
}
\end{minipage}
\begin{minipage}{0.5\textwidth}
\begin{subfigure}[b]{\textwidth}
\cstikz[0.9]{execution2.tikz}
\caption{A reduction sequence producing GHZ$_3$}
\label{sfig:ghzr}
\end{subfigure}
\end{minipage}
\caption{Example with lists of qubits and a recursive procedure.}
\label{fig:ghz}
\end{figure}

\begin{example}
The GHZ$_n$ state is defined as $\gamma_n \coloneqq (\ket{0}^{\otimes n} + \ket{1}^{\otimes n}) (\bra{0}^{\otimes n} + \bra{1}^{\otimes n}) / 2$.
In Figure~\ref{fig:ghz}, we define a procedure \texttt{GHZ}, which given a natural number $n$, generates the state $\gamma_n$, which is represented as a list of qubits of length $n$.
The procedure (\ref{sfig:ghz}) uses an auxiliary procedure \texttt{GHZnext}, which given a list of qubits representing the state $\gamma_n$, returns the state $\gamma_{n+1}$ again represented as a list of qubits.
The two procedures make use of some (hopefully obvious) syntactic sugar. In~\ref{sfig:ghzr}, we also present the last few steps of a reduction sequence which produces $\gamma_3$ starting
from configuration $(\texttt{l = GHZ(n)}\ |\ \texttt{n} = \texttt{s(s(s(zero)))}\ |\ \Omega\ |\ 1)$, where $\Omega$ contains the above mentioned procedures. In the reduction sequence we only show the term in evaluating position and we omit some intermediate steps.
The type \texttt{ListQ} is a shorthand for $\textbf{List}(\qbit)$ from Example~\ref{ex:syntax}.
\end{example}

\section{W*-algebras}\label{sec:model}
\newcommand{\mat}[1]{M_{#1}}

In this section we describe our denotational model. It is based on W*-algebras,
which are algebras of observables (i.e. physical entities), with interesting
domain-theoretic properties.  We recall some background on W*-algebras and
their categorical structure.
We refer the reader to \cite{takesaki} for an encyclopaedic account on W*-algebras.

\paragraph{Domain-theoretic preliminaries.}

\newcommand{\lub}{\bigvee}
Recall that a directed subset of a poset $P$ is a non-empty subset $X \subseteq P$
in which every pair of elements of $X$ has an upper bound in $X$.  A poset $P$
is a \emph{directed-complete partial order} (\textit{dcpo}) if each directed
subset has a supremum. A poset $P$ is \textit{pointed} if it has a least
element, usually denoted by $\perp$. A monotone map $f: P \to Q$ between posets
is \textit{Scott-continuous} if it preserves suprema of directed subsets. If
$P$ and $Q$ are pointed and $f$ preserves the least element, then we
say $f$ is \emph{strict}. We write $\cpo$ ($\cpobs$) for the category of (pointed) dcpo's and
(strict) Scott-continuous maps between them.

\paragraph{Definition of W*-algebras.}

\newcommand{\norm}[1]{\|#1\|}
A \emph{complex algebra} is a complex vector space $V$ equipped with a bilinear multiplication $( - \cdot - ): V \times V \to V$, which we write as juxtaposition.
A \emph{Banach algebra} $A$ is a complex algebra $A$ equipped with a submultiplicative norm $\norm{-} : A \to \mathbb R_{\geq 0}$, i.e.~$\forall x, y \in A: \norm{x y} \leq \norm{x}\norm{y}$.
A $\ast$-\emph{algebra} $A$ is a complex algebra $A$ with an involution $(-)^* : A \rightarrow A$ such that $(x^*)^*=x$, $(x+y)^*=(x^*+y^*)$, $(xy)^*=y^*x^*$ and $(\lambda x)^* = \overline{\lambda}x^* , $ for $x, y \in A$ and $\lambda \in \mathbb C$.
A \textit{C*-algebra} is a Banach $\ast$-algebra $A$ which satisfies the C*-identity, i.e.~$\norm{x^*x} = \norm{x}^2$ for all $x \in A$. A C*-algebra $A$ is \textit{unital} if it has an element $1 \in A,$ such that for every $x \in A : $ $x1=1x=x$.
All C*-algebras in this paper are unital and for brevity we regard unitality as part of their definition.

\newcommand{\N}{\mathbb{N}}

\begin{example}
\label{ex:c*-algebras}
The algebra $\mat n(\mathbb C)$ of $n\times n$ complex matrices is a C*-algebra. In particular, the set of complex numbers $\mathbb C$ has a C*-algebra structure since $\mat 1(\mathbb C) = \mathbb C $.
More generally, the $n\times n$ matrices valued in a C*-algebra $A$ also form a C*-algebra $\mat n(A)$. The C*-algebra of qubits is $\qbit \coloneqq \mat 2(\mathbb C).$
\end{example}

An element $x \in A$ of a C*-algebra $A$ is called \emph{positive} if $\exists y \in A : x=y^* y$. The \emph{poset of positive elements} of $A$ is denoted $A^+$ and its order is given by $x\leq y$ iff $(y-x) \in A^+$.
The \emph{unit interval} of $A$ is the subposet $[0,1]_A \subseteq A^+$ of all positive elements $x$ such that $0 \leq x \leq 1.$

Let $f: A \rightarrow B$ be a linear map between C*-algebras $A$ and $B$.
We say that $f$ is \textit{positive} if it preserves positive elements.
We say that $f$ is \textit{completely positive} if it is $n$-positive for every $n \in \N$, i.e. the map $\mat n(f):\mat n(A) \to \mat n(B)$ defined for every matrix $[x_{i,j}]_{1 \leq i,j\leq n} \in \mat n(A)$ by $\mat n(f)([x_{i,j}]_{1 \leq i,j\leq n})=[f(x_{i,j})]_{1 \leq i,j\leq n}$ is positive.
The map $f$ is called \textit{multiplicative}, \textit{involutive}, \textit{unital} if it preserves multiplication, involution, and the unit, respectively. The map $f$ is called \textit{subunital} whenever the inequalities $0 \leq f(1) \leq 1$ hold.
A \emph{state} on a C*-algebra $A$ is a completely positive unital map $s : A \to \mathbb C$.

\newcommand{\unit}{[0,1]}
Although W*-algebras are commonly defined in topological terms (as C*-algebras
closed under several operator topologies) or equivalently in algebraic terms (as
C*-algebras which are their own bicommutant), one can also equivalently define them in
domain-theoretic terms \cite{rennela-msc}, as we do next.

A completely positive map between C*-algebras is \emph{normal} if
its restriction to the unit interval is Scott-continuous~\cite[Proposition~A.3]{rennela-msc}.
A \emph{W*-algebra} is a C*-algebra $A$ such that the unit interval $\unit_A$ is a dcpo,
and $A$ has a separating set of normal states: for every $x \in
A^+$, if $x \neq 0$, then there is a normal state $s:A \to \mathbb C$ such that
$s(x)\neq 0$ \cite[Theorem~III.3.16]{takesaki}.

A linear map $f : A \to B$ between W*-algebras $A$ and $B$ is called an \emph{NCPSU-map} if $f$ is normal, completely positive and subunital. The map $f$ is called an \emph{NMIU-map} if $f$ is normal, multiplicative, involutive and unital.
We note that every NMIU-map is necessarily an NCPSU-map and that W*-algebras are closed under formation of matrix algebras as in Example~\ref{ex:c*-algebras}.

\paragraph{Categorical Structure.}

Let $\Wstar$ be the category of W*-algebras and NCPSU-maps and let $\WNMIU$ be its full-on-objects subcategory of NMIU-maps.
Throughout the rest of the paper let $\CC \coloneqq (\Wstar)^\op$ and let $\VV \coloneqq (\WNMIU)^\op.$
QPL types are interpreted as functors $\lrb{\Theta \vdash A} : \VV^{|\Theta|} \to \VV$ and closed QPL types as objects $\lrb A \in \Ob(\VV) = \Ob(\CC).$
One should think of $\VV$ as the category of \emph{values}, because the
interpretation of our values from~\secref{sec:operational} are indeed
$\VV$-morphisms. General QPL terms are interpreted as morphisms of $\CC$, so
one should think of $\CC$ as the category of \emph{computations}.
We now describe the categorical structure of $\VV$ and $\CC$ and later we justify our choice for working in the opposite categories.

Both $\CC$ and $\VV$ have a symmetric monoidal structure when equipped with the spacial tensor product, denoted here by $(- \otimes -)$, and tensor unit $I \coloneqq \mathbb C$~\cite[Section~10]{quantum-collections}. 
Moreover, $\VV$ is symmetric monoidal closed and also complete and cocomplete~\cite{quantum-collections}.
$\CC$ and $\VV$ have finite coproducts, given by direct sums of W*-algebras \cite[Proposition~4.7.3]{cho-msc}.
The coproduct
of objects $A$ and $B$ is denoted by $A+B$ and the coproduct injections are
denoted $\mathrm{left}_{A,B} : A \to A+B$ and $\mathrm{right}_{A,B}: B \to A+B.$
Given morphisms $f: A \to C$ and $g:B \to C$, we write
$[f,g]: A+B \to C$ for the unique cocone morphism induced by the coproduct.
Moreover, coproducts distribute over tensor products \cite[\S 4.6]{cho-msc}.
More specifically, there exists a natural isomorphism
$ d_{A,B,C}: A \otimes (B+C) \to (A \otimes B) + (A \otimes C) $ which
satisfies the usual coherence conditions.
The initial object in $\CC$ is moreover a zero object and is denoted $0.$ 
The W*-algebra of bits is $\textbf{bit} \coloneqq I + I = \mathbb C \oplus \mathbb C$.

The categories $\VV, \CC$ and $\Set$ are related by symmetric monoidal adjunctions:
\[
  \stikz{adjunctions.tikz}
\]
and the subcategory inclusion $J$ preserves coproducts and tensors up to equality.

Interpreting QPL within $\CC$ and $\VV$ is not an ad hoc trick. In physical terms, this corresponds to
adopting the \emph{Heisenberg picture} of quantum mechanics and this is usually
done when working with infinite-dimensional W*-algebras (like we do). Semantically, this is necessary, because (1)
our type system has conditional branching and we need to interpret QPL terms within a category with finite coproducts;
(2) we have to be able to compute parameterised initial algebras to interpret inductive datatypes.
The category $\Wstar$ has finite products, but it does \emph{not} have coproducts, so by interpreting QPL terms within $\CC = (\Wstar)^\op$ we solve problem (1). For (2), the monoidal closure of $\VV = (\WNMIU)^\op$ is crucial, because
it implies the tensor product preserves $\omega$-colimits.

\paragraph{Convex sums.}

In both $\CC$ and $\Wstar$, morphisms are closed under \emph{convex sums}, which are defined pointwise, as usual.
More specifically, given NCPSU-maps $f_1, \ldots, f_n : A \to B$ and real numbers $p_i \in [0,1]$ with $\sum_i p_i \leq 1,$ then
the map $\sum_i p_i f_i : A \to B$ is also an NCPSU-map.

\paragraph{$\cpobs$-enrichment.}

For W*-algebras $A$ and $B$, we define a partial order on $\CC(A,B)$ by :
$f \leq g$ iff $g-f$ is a completely positive map.
Equipped with this order, our category $\CC$ is $\cpobs$-enriched~\cite[Theorem~4.3]{cho-qpl-ng}.
The least element in $\CC(A,B)$ is also a zero morphism and is given by the map $\ZERO : A \to B$, defined by $\ZERO(x) = 0.$
Also, the coproduct structure and the symmetric monoidal structure are both $\cpobs$-enriched \cite[Corollary~4.9.15]{cho-msc}
\cite[Theorem~4.5]{cho-qpl-ng}.



\paragraph{Quantum Operations.}

For convenience, our operational semantics adopts the \emph{Schr{\"o}dinger
picture} of quantum mechanics, which is the picture most experts in quantum
computing are familiar with. However, as we have just explained, our
denotational semantics has to adopt the Heisenberg picture. The two pictures
are equivalent in finite dimensions and we will now show how to translate from
one to the other. By doing so, we provide an explicit description (in both
pictures) of the required quantum maps that we need to interpret QPL.

\begin{figure}[t]
\centerline{
$
\begin{array}{l|l|l|l}
    \mathrm{tr} : \mat n(\mathbb C) \to \mathbb C \ & \ \mathrm{new}_\rho : \mathbb C \to \mat{2^n}(\mathbb C) \ & \ \mathrm{meas} : \mat 2(\mathbb C) \to \mathbb C \oplus \mathbb C  \ & \ \mathrm{unitary}_S : \mat{2^n}(\mathbb C) \to \mat{2^n}(\mathbb C) \\
    \mathrm{tr} :: A \mapsto \sum_i A_{i,i} \ &
\ \mathrm{new}_\rho :: a \mapsto a \rho \ & \ \mathrm{meas} :: \begin{pmatrix} a & b \\ c & d \end{pmatrix} \mapsto \begin{pmatrix} a & d \end{pmatrix} \ & \ \mathrm{unitary}_S :: A \mapsto SAS^\dagger \\
    \mathrm{tr}^\dagger : \mathbb C \to \mat n(\mathbb C) \ &
\ \mathrm {new}_\rho^\dagger : \mat{2^n}(\mathbb C) \to \mathbb C \ & \ \mathrm{meas}^\dagger : \mathbb C \oplus \mathbb C  \to \mat 2(\mathbb C) \ & \ \mathrm{unitary}^\dagger_S : \mat{2^n}(\mathbb C) \to \mat{2^n}(\mathbb C) \\
    \mathrm{tr}^\dagger :: a \mapsto a I_n \ &
\ \mathrm{new}_\rho^\dagger :: A \mapsto \trace{A \rho} \ & \ \mathrm{meas}^\dagger :: \begin{pmatrix} a & d \end{pmatrix} \mapsto \begin{pmatrix} a & 0 \\ 0 & d \end{pmatrix} & \ \mathrm{unitary}_S^\dagger :: A \mapsto S^\dagger AS
\end{array}
$
}
\caption{A selection of maps in the Schr{\"o}dinger picture $(f: A \to B)$ and their Hermitian adjoints $(f^\dagger : B \to A)$.}
\label{fig:quantum-pictures}
\end{figure}

Consider the maps in Figure~\ref{fig:quantum-pictures}. The map \emph{tr} is
used to trace out (or discard) parts of quantum states. Density matrices $\rho$
are in 1-1 correspondence with the maps \emph{new}$_\rho$, which we use in our
semantics to describe (mixed) quantum states. The \emph{meas} map simply
measures a qubit in the computational basis and returns a bit as measurement
outcome. The \emph{unitary}$_S$ map is used for application of a unitary $S$. These
maps work as described in the Schr{\"o}dinger picture of quantum mechanics, i.e.,
the category $\Wstar$. For every map $f: A \to B$ among those mentioned,
$f^\dagger : B \to A$ indicates its Hermitian adjoint
\footnote{This adjoint exists, because $A$ and $B$ are \emph{finite-dimensional} W*-algebras which therefore have the structure of a Hilbert space when equipped with the Hilbert-Schmidt inner product~\cite[pp. 145]{bas-thesis}.}.
In the Heisenberg picture,
composition of maps is done in the opposite way, so we simply write $f^\ddagger
\coloneqq (f^\dagger)^\op \in \CC(A,B)$ for the Hermitian adjoint of $f$ when
seen as a morphism in $(\Wstar)^\op = \CC.$ Thus, the mapping $(-)^\ddagger$
translates the above operations from the Schr{\"o}dinger picture (the category $\Wstar$) to the
Heisenberg picture (the category $\CC$) of quantum mechanics.

\paragraph{Parameterised initial algebras.}
In order to interpret inductive datatypes, we need to be able to compute parameterised initial
algebras for the functors induced by our type expressions. $\VV$ is ideal for this, because
it is cocomplete and monoidal closed and so all type expressions induce functors on $\VV$ which preserve $\omega$-colimits.

\begin{definition}[cf. {\cite[\S 6.1]{fiore-thesis}}]
\label{def:param-initial}
  Given a category $\AAA$ and a functor $T : \AAA^n \to \AAA,$ with $n \geq 1$, a \emph{parameterised initial algebra}
  for $T$ is a pair $(T^\sharp, \phi^T),$ such that:
  \begin{itemize}
    \item $T^\sharp : \AAA^{n-1} \to \AAA$ is a functor;
    \item $\phi^T : T \circ \langle \Id, T^\sharp \rangle \naturalto T^\sharp : \AAA^{n-1} \to \AAA$ is a natural isomorphism;
    \item For every $A \in \Ob(\AAA^{n-1})$, the pair $(T^\sharp A, \phi^T_{A})$ is an initial $T(A, -)$-algebra.
  \end{itemize}
\end{definition}

\begin{proposition}
\label{prop:initial-algebra-exists}
Every $\omega$-cocontinuous functor $T : \VV^n \to \VV$ has a parameterised initial algebra $(T^\sharp, \phi^T)$. Moreover, $T^\sharp : \VV^{n-1} \to \VV$ is itself $\omega$-cocontinuous.
\end{proposition}
\begin{proof}
$\VV$ is cocomplete, so this follows from \cite[\S 4.3]{lnl-fpc-lmcs} (and also \cite[\S 4.3]{lnl-fpc}).
\qed\end{proof}

\section{Denotational Semantics of QPL}\label{sec:den}
In this section we describe the denotational semantics of QPL.

\subsection{Interpretation of types}\label{sub:types}
The interpretation of a type $\Theta \vdash A$ is a functor 
$\lrb{\Theta \vdash A} : \VV^{|\Theta|} \to \VV$, defined by induction on the derivation of $\Theta \vdash A$ in Figure~\ref{fig:type-interpretation}.
As usual, one has to prove this assignment is well-defined by showing the required initial algebras exist.
\begin{proposition}
The assignment in Figure~\ref{fig:type-interpretation} is well-defined.
\end{proposition}
\begin{proof}
By induction, every $\lrb{\Theta \vdash A}$ is an $\omega$-cocontinuous functor and thus it has a parameterised initial algebra by Proposition~\ref{prop:initial-algebra-exists}.
\qed\end{proof}

\begin{lemma}[Type Substitution]\label{lem:type-sub-simple}
Given types $\Theta, X \vdash A$ and $\Theta \vdash B$, then:
\[ \lrb{\Theta \vdash A[B/X]} = \lrb{\Theta, X \vdash A} \circ \langle \Id, \lrb{\Theta \vdash B}\rangle . \]
\end{lemma}
\begin{proof}
Straightforward induction.
\qed\end{proof}
For simplicity, the interpretation of terms is only defined on closed types and so we introduce more concise notation for them.
For any closed type $\cdot \vdash A$ we write for convenience 
$\lrb{A} \coloneqq \lrb{\cdot \vdash A}(*) \in \Ob(\VV), $ 
where $*$ is the unique object of the terminal category $\textbf{1}$.
Notice also that $\lrb A \in \Ob(\CC) = \Ob(\VV).$ 
\begin{definition}\label{def:fold-unfold-simple}
Given a closed type $\cdot \vdash \mu X. A$, we define an isomorphism (in $\VV$):
\begin{align*}
\sfold_{\mu X.A} &: \lrb{A[\mu X. A/ X]}  = \lrb{X \vdash A} \lrb{\mu X. A}   \cong \lrb{\mu X. A}  : \sunfold_{\mu X.A}
\end{align*}
where the equality is Lemma~\ref{lem:type-sub-simple} and the iso is the initial algebra structure.
\end{definition}

\begin{figure}[p]
\begin{minipage}{0.4\textwidth}
   \begin{align*}
    \lrb{\Theta \vdash A}           &: \VV^{|\Theta|} \to \VV \\
    \lrb{\Theta \vdash \Theta_i}    &= \Pi_i\\
    \lrb{\Theta \vdash I}           &= K_{I} \\
    \lrb{\Theta \vdash \qbit}       &= K_{\qbit} \\
    \lrb{\Theta \vdash A + B}       &= + \circ \langle \lrb{\Theta \vdash A} , \lrb{\Theta \vdash B} \rangle \\
    \lrb{\Theta \vdash A \otimes B} &= \otimes \circ \langle \lrb{\Theta \vdash A} , \lrb{\Theta \vdash B} \rangle \\
    \lrb{\Theta \vdash \mu X. A}    &= \lrb{\Theta, X \vdash A}^\sharp
    \end{align*}
\caption{Interpretations of types. $K_A$ is the constant-$A$-functor. }
\label{fig:type-interpretation}
\end{minipage}
\hfill
\begin{minipage}{0.4\textwidth}
\begin{align*}
  \lrb{\cdot \vdash * : I}                    &\coloneqq \id_I\\
  \lrb{ \{n\} \vdash n : \qbit}                    &\coloneqq \id_\qbit\\
  \lrb{Q \vdash \lleft_{A,B} v: A+B}          &\coloneqq \mathrm{left} \circ \lrb v\\
  \lrb{Q \vdash \rright_{A,B} v: A+B}         &\coloneqq \mathrm{right} \circ \lrb v\\
  \lrb{Q_1, Q_2 \vdash (v, w) : A \otimes B}  &\coloneqq \lrb v \otimes \lrb w\\
  \lrb{Q \vdash \fold_{\mu X.A} v : \mu X. A} &\coloneqq \mathrm{fold} \circ \lrb v
\end{align*}
\caption{Interpretation of values.}
\label{fig:value-interpretation}
\end{minipage}
\end{figure}
\begin{figure}[p]
$
\small{
\begin{array}{l}
\lrb{\Pi \vdash \langle \Gamma \rangle\  \newunit\ u\ \langle \Gamma, u:I \rangle} \coloneqq \pi \mapsto \left(
      \lrb{\Gamma}
        \xrightarrow{\cong}
      \lrb{\Gamma} \otimes I
      \right)\\
\lrb{\Pi \vdash \langle \Gamma, x : A \rangle\  \textbf{discard}\ x\ \langle \Gamma \rangle} \coloneqq \pi \mapsto \left(
      \lrb{\Gamma} \otimes \lrb A
        \xrightarrow{\id \otimes \diamond}
      \lrb{\Gamma} \otimes I
        \xrightarrow{\cong}
      \lrb{\Gamma}
      \right)\\
\lrb{\Pi \vdash \langle \Gamma, x : P \rangle\ y = \textbf{copy}\ x\ \langle \Gamma, x : P, y : P \rangle} \coloneqq \pi \mapsto \left(
      \lrb{\Gamma} \otimes \lrb P
        \xrightarrow{\id \otimes \triangle}
      \lrb{\Gamma} \otimes \lrb P \otimes \lrb P
      \right)\\
\lrb{\Pi \vdash \langle \Gamma \rangle\  \newqbit\ q\ \langle \Gamma, q:\qbit \rangle} \coloneqq \pi \mapsto \left(
      \lrb{\Gamma}
        \xrightarrow{\cong}
      \lrb{\Gamma} \otimes I
        \xrightarrow{\id \otimes \mathrm{new}_{\ket 0 \bra 0}^\ddagger}
      \lrb{\Gamma} \otimes \qbit 
      \right)\\
\lrb{\Pi \vdash \langle \Gamma, q: \qbit \rangle\ b = \textbf{measure}\ q\ \langle \Gamma, b: \textbf{bit} \rangle} \coloneqq \pi \mapsto \left(
      \lrb{\Gamma} \otimes \qbit
        \xrightarrow{\id \otimes \mathrm{meas}^\ddagger}
      \lrb{\Gamma} \otimes \textbf{bit}\right)\\
\lrb{\Pi \vdash \langle \Gamma, \vec q : \vec \qbit \rangle\ 
    \vec q \unitaryeqq S\ \langle \Gamma, \vec q : \vec \qbit \rangle} \coloneqq \pi \mapsto \left(
      \lrb{\Gamma} \otimes \qbit^{\otimes n} 
        \xrightarrow{\id \otimes \mathrm{unitary}_S^\ddagger}
      \lrb{\Gamma} \otimes \qbit^{\otimes n} 
      \right)\\
\lrb{\Pi \vdash \langle \Gamma \rangle\ M;N\ \langle \Sigma \rangle} \coloneqq \pi \mapsto \left(
      \lrb{\Gamma}
        \xrightarrow{\lrb{M}(\pi)}
      \lrb{\Gamma'}
        \xrightarrow{\lrb N(\pi)}
      \lrb{\Sigma}
      \right)\\
\lrb{\Pi \vdash \langle \Gamma \rangle\ \textbf{skip}\ \langle \Gamma \rangle} \coloneqq \pi \mapsto \left(
      \lrb{\Gamma}
        \xrightarrow{\id}
      \lrb{\Gamma}
      \right)\\
\lrb{\Pi \vdash \langle \Gamma, b: \textbf{bit} \rangle\ \textbf{while}\ b\ \textbf{do}\ {M}\ \langle \Gamma, b: \textbf{bit} \rangle} \coloneqq \pi \mapsto \left(
      \lrb{\Gamma} \otimes \textbf{bit} 
        \xrightarrow{\mathrm{lfp}(W_{\lrb M(\pi)})}
      \lrb{\Gamma} \otimes \textbf{bit} \right)\\
\lrb{\Pi \vdash \langle \Gamma, x:A \rangle\ y = \lleft_{A,B}\ x\ \langle \Gamma, y: A+B \rangle} \coloneqq \pi \mapsto \left(
      \lrb{\Gamma} \otimes \lrb A
        \xrightarrow{ \id \otimes \mathrm{left}_{A,B}}
      \lrb \Gamma \otimes (\lrb A + \lrb B) 
\right)\\
\lrb{\Pi \vdash \langle \Gamma, x:B \rangle\ y = \rright_{A,B}\ x\ \langle \Gamma, y: A+B \rangle} \coloneqq \pi \mapsto \left(
      \lrb{\Gamma} \otimes \lrb B
        \xrightarrow{ \id \otimes \mathrm{right}_{A,B}}
      \lrb \Gamma \otimes (\lrb A + \lrb B)
\right)\\
\llbracket \Pi \vdash \langle \Gamma, y: A+B \rangle\  \ccase\ y\ \textbf{of}\ \{\lleft\ x_1 \to M_1\ |\ \rright\ x_2 \to M_2\}\ \langle \Sigma \rangle \rrbracket \coloneqq \\
\qquad \qquad \pi \mapsto \left( \lrb \Gamma \otimes (\lrb{A} + \lrb B) \xrightarrow{d} (\lrb \Gamma \otimes \lrb A) + (\lrb \Gamma \otimes \lrb B) \xrightarrow{\left[\lrb {M_1}(\pi), \lrb {M_2}(\pi) \right]} \lrb \Sigma \right)\\
\lrb{\Pi \vdash \langle \Gamma, x_1: A, x_2: B \rangle\ x=(x_1, x_2) \ \langle \Gamma, x: A \otimes B \rangle} \coloneqq \pi \mapsto \left(
      \lrb{\Gamma} \otimes \lrb A \otimes \lrb B
        \xrightarrow{\id}
      \lrb{\Gamma} \otimes \lrb A \otimes \lrb B \right)\\
\lrb{\Pi \vdash \langle \Gamma, x: A \otimes B \rangle\ (x_1, x_2) = x \ \langle \Gamma, x_1: A, x_2: B \rangle} \coloneqq \pi \mapsto \left(
      \lrb{\Gamma} \otimes \lrb A \otimes \lrb B
        \xrightarrow{\id}
      \lrb{\Gamma} \otimes \lrb A \otimes \lrb B \right)\\
\lrb{\Pi \vdash \langle \Gamma, x: A[\mu X. A / X] \rangle\ y = \textbf{fold}\ x\ \langle \Gamma, y: \mu X.A \rangle} \coloneqq \pi \mapsto \left(
      \lrb{\Gamma} \otimes \lrb{A[\mu X. A / X]}
        \xrightarrow{\id \otimes \sfold}
      \lrb \Gamma \otimes \lrb{\mu X.A} \right)\\
\lrb{\Pi \vdash \langle \Gamma, x: \mu X. A \rangle\ y = \textbf{unfold}\ x\ \langle \Gamma, y: A[\mu X. A / X] \rangle} \coloneqq \pi \mapsto \left(
      \lrb{\Gamma} \otimes \lrb{\mu X. A}
        \xrightarrow{\id \otimes \sunfold}
      \lrb \Gamma \otimes \lrb{A[\mu X. A / X]} \right)\\
\lrb{\Pi \vdash \langle \Gamma \rangle\ \textbf{proc}\ f ::\ x:A \to y:B\ \{M\}\ \langle \Gamma \rangle} \coloneqq \pi \mapsto\ \left( \lrb{\Gamma} \xrightarrow{\id} \lrb{\Gamma} \right)\\
\lrb{\Pi, f: A \to B \vdash \langle \Gamma, x: A \rangle\ y = f(x)\ \langle \Gamma, y:B \rangle} \coloneqq (\pi, f) \mapsto \left( \lrb \Gamma \otimes \lrb A \xrightarrow{ \id \otimes f} \lrb \Gamma \otimes \lrb B \right)
\end{array}
}
$
\caption{Interpretation of QPL terms.}
\label{fig:semantics}
\end{figure}

\begin{example}
The interpretation of the types from Example~\ref{ex:syntax} are $\lrb{\mathbf{Nat}} =  \bigoplus_{i=0}^\omega \mathbb C$
and $\lrb{\mathbf{List}(A)} = \bigoplus_{i=0}^\omega {\lrb A}^{\otimes i}$.
Specifically, $\lrb{\mathbf{List}(\mathbf{qbit})} = \bigoplus_{i=0}^\omega {\mathbb C^{2^i \times 2^i}}.$
\end{example}

\subsection{Copying and Discarding}\label{sub:copy-discard}

Our type system is affine, so we have to construct discarding maps at all types. The tensor unit $I$ is a terminal object in $\VV$ (but not in $\CC$) which leads us to the next definition.
\begin{definition}[Discarding map]
\label{def:discarding-map}
For any W*-algebra $A$, let $\diamond_A : A \to I$ be the unique morphism of $\VV$ with the indicated domain and codomain.
\end{definition}
We will see that all values admit an interpretation as $\VV$-morphisms and are therefore discardable. In physical terms, this means values are causal (in the sense mentioned in the introduction). 
Of course, this is not true for the interpretation of general terms (which correspond to $\CC$-morphisms).

Our language is equipped with a copy operation on classical data, so we have to explain how to copy classical values. In the absence of a !-modality, we
do this by constructing a copy map defined at all \emph{classical} types.
\begin{proposition}
\label{prop:comono}
Using the categorical data of \stikz{structural-adjunction.tikz}, one can define a copy map $\triangle_{\lrb P} : \lrb P \to \lrb P \otimes \lrb P$ for every classical type $\cdot \vdash P$, such that
the triple $\left( \lrb P, \triangle_{\lrb P}, \diamond_{\lrb P} \right)$ forms a cocommutative comonoid in $\VV$.
\end{proposition}
\begin{proof}
The proof requires some effort and is presented in Appendix~\ref{app:comonoids}.
\qed\end{proof}
We shall later see that the interpretations of our \emph{classical} values are comonoid homomorphisms (w.r.t. Proposition \ref{prop:comono}) and therefore they may be copied. 

\subsection{Interpretation of terms}\label{sub:terms}
Given a variable context $\Gamma = x_1:A_1, \ldots, x_n:A_n,$ we interpet it as the object $\lrb{\Gamma} := \lrb{A_1} \otimes \cdots \otimes \lrb{A_n} \in \text{Ob}(\CC).$
The interpretation of a procedure context $\Pi = f_1:A_1 \to B_1, \ldots, f_n:A_n \to B_n$ is defined to be
the pointed dcpo $\lrb{\Pi} := \CC(A_1,B_1) \times \cdots \times \CC(A_n,B_n).$
A term $\Pi \vdash \langle \Gamma \rangle\ M\ \langle \Sigma \rangle$ is interpreted as a Scott-continuous function
$ \lrb{\Pi \vdash \langle \Gamma \rangle\ M\ \langle \Sigma \rangle} : \lrb{\Pi} \to \CC(\lrb \Gamma, \lrb \Sigma) $
defined by induction on the derivation of $\Pi \vdash \langle \Gamma \rangle\ M\ \langle \Sigma \rangle$ in Figure~\ref{fig:semantics}.
For brevity, we often write $\lrb{M} \coloneqq \lrb{\Pi \vdash \langle \Gamma \rangle\ M\ \langle \Sigma \rangle}$, when the contexts are clear or unimportant.

We now explain some of the notation used in Figure~\ref{fig:semantics}.
The rules for manipulating qubits use the morphisms $\mathrm{new}_{\ket 0 \bra 0}^\ddagger, \mathrm{meas}^\ddagger$ and $\mathrm{unitary}_S^\ddagger$ which are defined in \secref{sec:model}.
For the interpretation of \textbf{while} loops,
given an arbitrary morphism $f:{A \otimes \bit \to A \otimes \bit}$ of $\CC$, we define a Scott-continuous endofunction
\begin{align*}
W_{f} &: \CC \left(A \otimes \bit, A \otimes \bit \right) \to \CC(A \otimes \bit, A \otimes \bit)\\
W_{f}(g) &= \left[ \id \otimes \mathrm{left}_{I,I},\ g \circ f \circ (\id \otimes \mathrm{right}_{I,I}) \right] \circ d_{A,I,I} ,
\end{align*}
where the isomorphism $d_{A,I,I} : A \otimes (I+I) \to (A \otimes I) + (A \otimes I)$ is explained in~\secref{sec:model}.
For any pointed dcpo $D$ and Scott-continuous function $h : D \to D$, its \emph{least fixpoint} is $\mathrm{lfp}(h) \coloneqq \bigvee_{i=0}^\infty h^i(\perp),$ where $\perp$ is the least element of $D$.

\begin{remark}\label{rem:procedures}
The term semantics for defining and calling procedures does not
involve any fixpoint computations. The required fixpoint computations are
done when interpreting procedure stores, as we shall see next.
\end{remark}

\subsection{Interpretation of configurations}\label{sub:configurations}
Before we may interpret program configurations, we first have to describe how
to interpret values and procedure stores.

\paragraph{Interpretation of values.}
A qubit pointer context $Q$ is interpreted as the object $\lrb Q = \qbit^{\otimes |Q|}.$
A value $Q \vdash v : A$ is interpreted as a morphism in $\VV$
$\lrb{Q \vdash v: A} : \lrb Q \xrightarrow{} \lrb A,$ which we abbreviate as
$\lrb v$ if $Q$ and $A$ are clear from context. It is defined by induction
on the derivation of $Q \vdash v:A$ in Figure~\ref{fig:value-interpretation}.

For the next theorem, recall that if $Q \vdash v :A$ is a classical value, then $Q = \cdot.$
\begin{theorem}\label{thm:causal-values}
Let $Q \vdash v : A$ be a value. Then:
\begin{enumerate}
\item $\lrb{ v }$ is discardable (i.e. causal). More specifically, $\diamond_{\lrb A} \circ \lrb{v} = \diamond_{\lrb Q} = \mathrm{tr}^\ddagger.$
\item If $A$ is classical, then $\lrb{v}$ is copyable, i.e., $\triangle_{\lrb A} \circ \lrb{v} = (\lrb v \otimes \lrb v) \circ \triangle_{I}.$
\end{enumerate}

\end{theorem}
\begin{proof}
In Appendix~\ref{app:comonoids}.
\qed\end{proof}

We see that, as promised, interpretations of values may always be discarded and interpretations of classical values may also be copied. Next, we explain how to interpret value contexts.
For a value context $ Q ; \Gamma \vdash V$, its interpretation is the morphism:
\[ \lrb{ Q ; \Gamma \vdash V} = \left( \lrb{Q} \xrightarrow{\cong} \lrb{Q_1} \otimes \cdots \otimes \lrb{Q_n} \xrightarrow{\lrb{v_1} \otimes \cdots \otimes \lrb{v_n}} \lrb \Gamma \right) ,\]
where $Q_i \vdash v_i : A_i$ is the splitting of $Q$ (see~\secref{sec:operational}) and $\lrb \Gamma = \lrb{A_1} \otimes \cdots \otimes \lrb{A_n}.$
Some of the $Q_i$ can be empty and this is the reason why the definition depends on a coherent natural isomorphism.
We write $\lrb{V}$ as a shorthand for $\lrb{ Q ; \Gamma \vdash V}$. Obviously, $\lrb V$ is also causal thanks to Theorem~\ref{thm:causal-values}.

\paragraph{Interpretation of procedure stores.}
The interpretation of a well-formed procedure store $\Pi \vdash \Omega$ is an element of $\lrb{\Pi}$, i.e. a $|\Pi|$-tuple of morphisms from $\CC$. It is defined by induction on $\Pi \vdash \Omega:$
\begin{align*}
\lrb{\cdot \vdash \cdot} &= () & \lrb{\Pi, f : A \to B \vdash \Omega, f :: x:A \to y:B\ \{M\}} &= (\lrb \Omega, \mathrm{lfp}(\lrb{M}(\lrb \Omega, -))) .
\end{align*}

\paragraph{Interpretation of configurations.}
Density matrices $\rho \in M_{2^n}(\mathbb C)$ are in 1-1 correspondence with $\Wstar$-morphisms $\mathrm{new}_\rho : \mathbb C \to M_{2^n}(\mathbb C)$ which are in turn in 1-1 correspondence
with $\CC$-morphisms $\mathrm{new}_\rho^\ddagger : I \to \qbit^{\otimes n}$. 
Using this observation, we can now define the interpretation of a configuration $\mathcal C = (M\ |\ V\ |\ \Omega\ |\ \rho) $ with $\Pi; \Gamma; \Sigma; Q \vdash (M\ |\ V\ |\ \Omega\ |\ \rho)$ to be the morphism
\begin{align*}
&\lrb{\Pi; \Gamma; \Sigma; Q \vdash (M\ |\ V\ |\ \Omega\ |\ \rho)} = \left( I \xrightarrow{\mathrm{new}_\rho^\ddagger} \qbit^{\otimes \dim(\rho)} \xrightarrow{\lrb{  Q ; \Gamma \vdash V}} \lrb \Gamma \xrightarrow{\lrb{\Pi \vdash \langle \Gamma \rangle\ M\ \langle \Sigma \rangle}(\lrb{\Pi \vdash \Omega})} \lrb \Sigma \right) .
\end{align*}
For brevity, we simply write $\lrb{(M\ |\ V\ |\ \Omega\ |\ \rho)}$ or even just $\lrb{\mathcal C}$ to refer to the above morphism.

\subsection{Soundness, Adequacy and Big-step Invariance}\label{sub:soundness}

Since our operational semantics allows for branching, \emph{soundness} is showing that the interpretation of configurations is equal to the sum of small-step reducts.

\begin{theorem}[Soundness]\label{thm:soundness}
For any non-terminal configuration $\mathcal C $ :
\[ \lrb{\mathcal C} = \inlinesum_{\mathcal C \leadsto \mathcal D} \lrb{\mathcal D}. \]
\end{theorem}
\begin{proof}
By induction on the shape of the term component of $\mathcal C$.
\qed\end{proof}

\begin{remark}
The above sum and all sums that follow are well-defined convex sums of NCPSU-maps where the probability weights $p_i$ have been encoded in the density matrices.
\end{remark}

A natural question to ask is whether $\lrb{\mathcal C}$ is also equal to the (potentially infinite) sum of all terminal configurations that $\mathcal C$ reduces to.
In other words, is the interpretation of configurations also invariant with respect to big-step reduction.
This is indeed the case and proving this requires considerable effort.

\begin{theorem}[Big-step Invariance]\label{thm:big-step-invariance}
For any configuration $\mathcal C$, we have:
\[ \lrb{\mathcal C} = \bigvee_{n=0}^\infty \sum_{r \in \TerSeq_{\leq n}(\mathcal C)} \lrb{\End(r)} \]
\end{theorem}
\begin{proof}
In Appendix~\ref{app:big-step}.
\qed\end{proof}

The above theorem is the main result of our paper. This is a powerful result, because
with big-step invariance in place, computational adequacy at all types is now a simple consequence of the causal properties of our interpretation.
Observe that for any configuration $\mathcal C,$ we have a subunital map $\diamond \circ \lrb{\mathcal C} : \mathbb C \to \mathbb C$
and evaluating it at 1 yields a real number $\left( \diamond \circ \lrb{\mathcal C} \right)(1) \in [0,1]$.

\begin{theorem}[Adequacy]\label{thm:adequacy}
For any normalised $\mathcal C$ :
$ \left( \diamond \circ \lrb{\mathcal C} \right) \left( 1 \right) = \Halt(\mathcal C) . $
\end{theorem}
\begin{proof}
In Appendix~\ref{app:adequacy}.
\qed\end{proof}

If $\mathcal C $ is not normalised, then adequacy can be recovered simply by normalising:
$ \left( \diamond \circ \lrb{\mathcal C} \right) \left( 1 \right) = \trace{\mathcal C} \Halt(\mathcal C) , $ for any possible configuration $\mathcal C$.
The adequacy formulation of~\cite{quantitative} and~\cite{quantum-game} is now a special case of our more general formulation.

\begin{corollary}
Let $M$ be a closed program of unit type, i.e. $\cdot \vdash \langle \cdot \rangle\ M\ \langle \cdot \rangle.$ Then:
\[ \lrb{(M\ |\ \cdot\ |\ \cdot\ |\ 1)} \left( 1 \right) = \Halt(M\ |\ \cdot\ |\ \cdot\ |\ 1) . \]
\end{corollary}
\begin{proof}
By Theorem~\ref{thm:adequacy} and because $\diamond_I = \id$.
\qed\end{proof}

\section{Conclusion and Related Work}

There are many quantum programming languages described in the literature. For a
survey see~\cite{quantum-survey} and~\cite[pp. 129]{dagstuhl-survey}. Circuit
programming languages, such as Proto-Quipper~\cite{pqm,ross-pq} and
ECLNL~\cite{eclnl}, generate quantum circuits, but these
languages do not support executing quantum measurements. Here we focus on
quantum languages which support measurement and which have either inductive
datatypes or some computational adequacy result. 

Our work is the first to present a detailed semantic treatment of user-defined
inductive datatypes for quantum programming. In~\cite{quantitative}
and~\cite{quantum-game}, the authors show how to interpret a quantum lambda
calculus extended with a datatype for lists, but their syntax does not support
any other inductive datatypes. These languages are equipped with lambda
abstractions, whereas our language has only support for procedures. Lambda
abstractions are modelled using constructions from quantitative semantics of
linear logic in~\cite{quantitative} and techniques from game semantics
in~\cite{quantum-game}. We believe our model is simpler and
certainly more physically natural, because we work only with mathematical
structures used by physicists in their study of quantum mechanics.
Both~\cite{quantitative} and~\cite{quantum-game} prove an adequacy result for programs of unit type.
In~\cite{ewire}, the authors discuss potential categorical models for inductive datatypes in quantum programming, but there is no detailed semantic treatment provided and there is no adequacy result, because the language lacks recursion.

Other quantum programming languages without inductive datatypes, but which
prove computational adequacy results
include~\cite{quantum-goi,geometry-parallelism}.
A simple quantum language described in~\cite{ying-book}
is equipped with a denotational semantics which is defined in terms of its
operational semantics in such a way that soundness and adequacy trivially hold
by definition (and are indeed not stated).

A model based on W*-algebras for a quantum lambda calculus without recursion or
inductive datatypes was described in a recent manuscript~\cite{kenta-bram}. In
that model, it appears that currying is \emph{not} a Scott-continuous operation, and if so, the addition
of recursion renders the model neither sound, nor adequate.
For this reason, we use procedures and not lambda abstractions in our language.

To conclude, we presented two novel results in quantum programming: (1) we
provided a denotational semantics for a quantum programming language with inductive datatypes; (2)
we proved that our denotational semantics is invariant with respect to big-step
reduction. We also showed that the latter result is quite powerful by
demonstrating how it immediately implies computational adequacy.

Our denotational model is based on W*-algebras, which are used by physicists to
study quantum foundations.  We hope this would make it useful for developing
static analysis methods based on abstract interpretation and we plan on
investigating this as part of future work.


\newpage
\bibliography{refs}

\begin{thebibliography}{10}
\providecommand{\url}[1]{\texttt{#1}}
\providecommand{\urlprefix}{URL }
\providecommand{\doi}[1]{https://doi.org/#1}

\bibitem{cho-msc}
Cho, K.: Semantics for a quantum programming language by operator algebras
  (2014), master Thesis, University of Tokyo.

\bibitem{cho-qpl-ng}
Cho, K.: Semantics for a quantum programming language by operator algebras. New
  Generation Comput.  \textbf{34}(1-2),  25--68 (2016).
  \doi{10.1007/s00354-016-0204-3}

\bibitem{kenta-bram}
Cho, K., Westerbaan, A.: Von neumann algebras form a model for the quantum
  lambda calculus. CoRR  \textbf{abs/1603.02133} (2016),
  \url{http://arxiv.org/abs/1603.02133}

\bibitem{quantum-game}
Clairambault, P., de~Visme, M., Winskel, G.: Game semantics for quantum
  programming. {PACMPL}  \textbf{3}({POPL}),  32:1--32:29 (2019).
  \doi{10.1145/3290345}

\bibitem{fiore-thesis}
Fiore, M.P.: Axiomatic domain theory in categories of partial maps. Ph.D.
  thesis, University of Edinburgh, {UK} (1994)

\bibitem{quantum-survey}
Gay, S.J.: Quantum programming languages: survey and bibliography. Mathematical
  Structures in Computer Science  \textbf{16}(4),  581--600 (2006).
  \doi{10.1017/S0960129506005378}

\bibitem{grover}
Grover, L.K.: A fast quantum mechanical algorithm for database search. In:
  Proceedings of the Twenty-Eighth Annual {ACM} Symposium on the Theory of
  Computing, Philadelphia, Pennsylvania, USA, May 22-24, 1996. pp. 212--219
  (1996). \doi{10.1145/237814.237866}

\bibitem{quantum-goi}
Hasuo, I., Hoshino, N.: Semantics of higher-order quantum computation via
  geometry of interaction. Ann. Pure Appl. Logic  \textbf{168}(2),  404--469
  (2017). \doi{10.1016/j.apal.2016.10.010}

\bibitem{aleks-sander}
Kissinger, A., Uijlen, S.: A categorical semantics for causal structure. In:
  32nd Annual {ACM/IEEE} Symposium on Logic in Computer Science, {LICS} 2017,
  Reykjavik, Iceland, June 20-23, 2017. pp. 1--12. {IEEE} Computer Society
  (2017). \doi{10.1109/LICS.2017.8005095}

\bibitem{quantum-collections}
Kornell, A.: Quantum collections. International Journal of Mathematics
  \textbf{28}(12),  1750085 (2017). \doi{10.1142/S0129167X17500859}

\bibitem{geometry-parallelism}
Lago, U.D., Faggian, C., Valiron, B., Yoshimizu, A.: The geometry of
  parallelism: classical, probabilistic, and quantum effects. In: Castagna, G.,
  Gordon, A.D. (eds.) Proceedings of the 44th {ACM} {SIGPLAN} Symposium on
  Principles of Programming Languages, {POPL} 2017, Paris, France, January
  18-20, 2017. pp. 833--845. {ACM} (2017),
  \url{http://dl.acm.org/citation.cfm?id=3009859}

\bibitem{lnl-fpc-lmcs}
Lindenhovius, B., Mislove, M., Zamdzhiev, V.: Lnl-fpc: The linear/non-linear
  fixpoint calculus \url{https://arxiv.org/abs/1906.09503}, submitted.

\bibitem{lnl-fpc}
Lindenhovius, B., Mislove, M., Zamdzhiev, V.: Mixed linear and non-linear
  recursive types. Proc. ACM Program. Lang.  \textbf{3}(ICFP),  111:1--111:29
  (Jul 2019). \doi{10.1145/3341715}

\bibitem{eclnl}
Lindenhovius, B., Mislove, M.W., Zamdzhiev, V.: Enriching a linear/non-linear
  lambda calculus: {A} programming language for string diagrams. In: Dawar, A.,
  Gr{\"{a}}del, E. (eds.) Proceedings of the 33rd Annual {ACM/IEEE} Symposium
  on Logic in Computer Science, {LICS} 2018, Oxford, UK, July 09-12, 2018. pp.
  659--668. {ACM} (2018). \doi{10.1145/3209108.3209196}

\bibitem{quantum-algorithms-survey}
Mosca, M.: Quantum algorithms. In: Meyers, R.A. (ed.) Encyclopedia of
  Complexity and Systems Science, pp. 7088--7118. Springer (2009).
  \doi{10.1007/978-0-387-30440-3\_423}

\bibitem{dagstuhl-survey}
Mosca, M., Roetteler, M., Selinger, P.: {Quantum Programming Languages
  (Dagstuhl Seminar 18381)}. Dagstuhl Reports  \textbf{8}(9),  112--132 (2019).
  \doi{10.4230/DagRep.8.9.112}

\bibitem{quantitative}
Pagani, M., Selinger, P., Valiron, B.: Applying quantitative semantics to
  higher-order quantum computing. In: Jagannathan, S., Sewell, P. (eds.) The
  41st Annual {ACM} {SIGPLAN-SIGACT} Symposium on Principles of Programming
  Languages, {POPL} '14, San Diego, CA, USA, January 20-21, 2014. pp. 647--658.
  {ACM} (2014). \doi{10.1145/2535838.2535879}

\bibitem{rennela-msc}
Rennela, M.: {Operator Algebras in Quantum Computation} (2013), master Thesis,
  Université Paris 7 Denis Diderot

\bibitem{ewire}
Rennela, M., Staton, S.: Classical control and quantum circuits in enriched
  category theory. Electr. Notes Theor. Comput. Sci.  \textbf{336},  257--279
  (2018). \doi{10.1016/j.entcs.2018.03.027}

\bibitem{pqm}
Rios, F., Selinger, P.: A categorical model for a quantum circuit description
  language. In: {QPL} 2017 (2017). \doi{10.4204/EPTCS.266.11}

\bibitem{ross-pq}
Ross, N.J.: Algebraic and logical methods in quantum computation (2015), ph.D.
  thesis, Dalhousie University

\bibitem{qpl}
Selinger, P.: Towards a quantum programming language. Mathematical Structures
  in Computer Science  \textbf{14}(4),  527--586 (2004).
  \doi{10.1017/S0960129504004256}

\bibitem{shor}
Shor, P.W.: Polynomial-time algorithms for prime factorization and discrete
  logarithms on a quantum computer. {SIAM} Review  \textbf{41}(2),  303--332
  (1999). \doi{10.1137/S0036144598347011}

\bibitem{takesaki}
Takesaki, M.: {Theory of operator algebras. Vol. I, II and III}.
  Springer-Verlag, Berlin (2002)

\bibitem{wstar-adjoint}
Westerbaan, A.: Quantum programs as kleisli maps. In: Duncan, R., Heunen, C.
  (eds.) Proceedings 13th International Conference on Quantum Physics and
  Logic, {QPL} 2016, Glasgow, Scotland, 6-10 June 2016. {EPTCS}, vol.~236, pp.
  215--228 (2016). \doi{10.4204/EPTCS.236.14}

\bibitem{bas-thesis}
Westerbaan, B.: Dagger and dilations in the category of von Neumann algebras.
  Ph.D. thesis, Radboud University (2018),
  \url{http://arxiv.org/abs/1803.01911}

\bibitem{no-cloning}
Wootters, W.K., Zurek, W.H.: A single quantum cannot be cloned. Nature
  \textbf{299}(5886),  802--803 (1982)

\bibitem{ying-book}
Ying, M.: Foundations of Quantum Programming. Morgan Kaufmann (2016)

\end{thebibliography}

\newpage
\appendix

\appendix

\section{Comonoid Structure of Classical Types}\label{app:comonoids}
In this appendix we describe the comonoid structure of classical types within
our category $\VV$ (and therefore also in $\CC$). Our methods are based on those of~\cite{lnl-fpc}, but there
are some small differences compared to that work. In particular, we have less
structure to work with, but our type expressions are also simpler.
Nevertheless, the main idea is the same: we present an additional (classical) type
interpretation for classical types within a cartesian category (in our case
$\Set$) which then allows us to carry the comonoid structure into $\VV$ via a
symmetric monoidal adjunction. In our case, the adjunction is denoted
\stikz{lnl-fpc-model2.tikz} and it is defined in~\cite{kenta-bram}, but here we only use the categorical properties mentioned in \secref{sec:model}, so we omit the definition of the adjunction.

\subsection{Classical Interpretation of Classical Types}

The classical interpretation of a classical type $\Theta \vdash P$ is given by a functor 
$\flrb{\Theta \vdash P} : \SET^{|\Theta|} \to \SET$, defined by induction on the derivation of $\Theta \vdash P$ in the following way:
    \begin{align*}
    \flrb{\Theta \vdash P}           &: \SET^{|\Theta|} \to \SET \\
    \flrb{\Theta \vdash \Theta_i}    &= \Pi_i \\
    \flrb{\Theta \vdash I}           &= K_{1} \\
    \flrb{\Theta \vdash P + R}       &= \amalg \circ \langle \flrb{\Theta \vdash P} , \flrb{\Theta \vdash R} \rangle \\
    \flrb{\Theta \vdash P \otimes R} &= \times \circ \langle \flrb{\Theta \vdash P} , \flrb{\Theta \vdash R} \rangle \\
    \flrb{\Theta \vdash \mu X. P}    &= \flrb{\Theta, X \vdash P}^\sharp ,
    \end{align*}
where $\times$ is the categorical product in $\Set$ and where $\amalg$ is the categorical coproduct in $\Set$.
This assignment is well-defined, because every $\flrb{\Theta \vdash P}$ is an $\omega$-cocontinuous functor on $\SET$ and thus it has a parameterised initial algebra \cite[\S 4.3]{lnl-fpc}.
For any closed classical type $\cdot \vdash P$, we write
$ \flrb{P} \coloneqq \flrb{\cdot \vdash P}(*) \in \Ob(\SET).$

Our next proposition shows that the standard and classical type interpretations are strongly related via a natural isomorphism, which is crucial for defining the comonoid structure.

\begin{proposition}
  For any classical type $\Theta \vdash P$
  there exists a natural isomorphism
  \[ \iota^{\Theta \vdash P} : \lrb{\Theta \vdash P} \circ F^{\times |\Theta|}\cong F \circ \flrb{\Theta \vdash P} \qquad \text{(see Figure~\ref{fig:type-relationship}),} \]
  defined by induction on the derivation of $\Theta \vdash P.$
  Therefore, in the special case when $\Theta = \cdot$, there exists an isomorphism $\iota^P : \lrb P \cong F \flrb P$ given by $\iota^P \coloneqq \iota^{\cdot \vdash P}_{*}.$
  \begin{figure}[t]
    \centering
    \stikz{types-relation.tikz}
    \caption{Relationship between type interpretations.}
    \label{fig:type-relationship}
  \end{figure}
\end{proposition}
\begin{proof}
The definition of $\iota^{\Theta \vdash P}$ and the proof of the theorem is essentially the same as~\cite[Theorem 6.1.2]{lnl-fpc} and it is presented in detail in \cite[A.8]{lnl-fpc-lmcs}.
\qed\end{proof}

\subsection{Constructing the copy map}

Using $\iota^P$, we can define a copy map on the interpretation of closed classical types.

\begin{definition}\label{def:copy-map}
For any closed classical type $\cdot \vdash P$ we define a copy morphism: 
\begin{align*}
  \triangle_{\lrb P} &\coloneqq \left( \lrb{P} \xrightarrow{\iota} F \flrb{P} \xrightarrow{F \langle \id, \id \rangle } F \left( \flrb{P} \times \flrb{P} \right) \xrightarrow{\cong} F \flrb{P} \otimes F \flrb{P}
    \xrightarrow{\iota^{-1} \otimes \iota^{-1}} \lrb{P} \otimes \lrb{P} \right) .
\end{align*}
\end{definition}

\begin{proposition}\label{prop:comonoid}
$\left( \lrb P, \triangle_{\lrb P}, \diamond_{\lrb P} \right)$ is a cocommutative comonoid in $\VV$.
\end{proposition}
\begin{proof}
First, observe that 
  $ \diamond_{\lrb P} = \left( \lrb{P} \xrightarrow{\iota} F \flrb{P} \xrightarrow{F1} F1 \xrightarrow{\cong} I \right),$ because $I$ is terminal in $\VV$.
The rest of the proof is then identical to \cite[Proposition 6.3.3]{lnl-fpc}.
\qed\end{proof}

Next, we identify the comonoid homomorphisms with respect to the above comonoid structure.

\begin{definition}\label{def:comonoid-homomorphism}
Given closed classical types $P$ and $R$, we say that a morphism $f : \lrb P \to \lrb R$ is \emph{classical} if $f = \left( \lrb P \xrightarrow \iota F \flrb P \xrightarrow{Ff'} F \flrb R \xrightarrow{\iota^{-1}} \lrb R \right)$,
for some $f' : \flrb P \to \flrb R$ in $\Set.$
\end{definition}

\begin{proposition}\label{prop:comonoid-homomorphism}
For every classical morphism $f : \lrb P \to \lrb R $, we have:
\begin{align*}
\triangle_{\lrb R} \circ f = (f \otimes f) \circ \triangle_{\lrb P}  \qquad \text{ and } \qquad \diamond_{\lrb R} \circ f  = \diamond_{\lrb P}.
\end{align*}
Therefore, $f$ is a comonoid homomorphism with respect to Proposition~\ref{prop:comonoid}.
\end{proposition}
\begin{proof}
Essentially the same as \cite[Proposition 6.3.6]{lnl-fpc}.
\qed\end{proof}

\subsection{Folding and Unfolding of Classical Types}

In order to prove soundness of our semantics, we have to be able to copy all
classical values. Thus, we have to show that folding and unfolding of classical
inductive types are classical morphisms in the sense of
Definition~\ref{def:comonoid-homomorphism}. This is the purpose of the present subsection.

\begin{lemma}[Type Substitution]
\label{lem:classical-sub}
Let $\Theta, X \vdash P$ and $\Theta \vdash R$ be classical types.
Then (see Figure~\ref{fig:substitution}):
\begin{enumerate}
  \item $\flrb{\Theta \vdash P[R/X]} = \flrb{\Theta, X \vdash P} \circ \langle \Id, \flrb{\Theta \vdash R} \rangle ;$
  \item  $ \iota^{\Theta \vdash P[R/X]} = \iota^{\Theta, X \vdash P} \langle \Id, \flrb{\Theta \vdash R}\rangle \circ \lrb{\Theta, X\vdash P} \langle F^{\times |\Theta|}, \iota^{\Theta \vdash R} \rangle$
\end{enumerate}
\end{lemma}
\begin{proof}
Special case of \cite[Lemma 6.1.5]{lnl-fpc}, where $\gamma = \id$. Detailed proof is available in~\cite[A.9]{lnl-fpc-lmcs}.
\qed\end{proof}

We may now define folding and unfolding for the classical interpretation of classical inductive types.

\begin{definition}\label{def:fold-unfold-classical}
Given a closed classical type $\cdot \vdash \mu X. P,$ we define an isomorphism:
\begin{align*}
\ifold_{\mu X.P} &: \flrb{P[\mu X. P/ X]} = \flrb{X \vdash P} \flrb{\mu X. P} \cong \flrb{\mu X. P} : \iunfold_{\mu X.P}
\end{align*}
\end{definition}

Like in the standard type interpretation, substitution holds up to equality, so the above folding/unfolding
is given simply by the initial algebra structure of the
indicated functors. We conclude the subsection with a proposition showing
how the different folds/unfolds relate to one another.

\begin{proposition}\label{prop:fold-unfold-relationship}
Given a closed classical type $\cdot \vdash \mu X. P,$ then
\[ \sfold_{\mu X.P} = (\iota^{\mu X. P})^{-1} \circ F\ifold_{\mu X.P} \circ \iota^{P [\mu X.P / X]} \]
\text{(see Figure~\ref{fig:folds}).}
\begin{figure}[t]
\cstikz{fold-unfold.tikz}
\caption{Relationships between the different fold/unfold maps.}
\label{fig:folds}
\end{figure}
\end{proposition}
\begin{proof}
This is simply a special case of \cite[Theorem 6.1.7]{lnl-fpc}.
\qed\end{proof}

This shows that
folding/unfolding of classical types is a classical isomorphism
(Definition~\ref{def:comonoid-homomorphism}) and may thus be copied.

\newpage
\subsection{Copying of Classical Values}

Finally, we may state the main theorem from this appendix.

\begin{figure}[t]
	\cstikz{alpha-substitution.tikz}
	\caption{The commuting diagram of natural isomorphisms for Lemma~\ref{lem:classical-sub}.}\label{fig:substitution}
\end{figure}

\begin{theorem}
The interpretation of a classical value is a classical morphism (see Definition~\ref{def:comonoid-homomorphism}) and may thus be copied.
\end{theorem}
\begin{proof}
The proof is essentially the same as~\cite[Proposition 6.3.7]{lnl-fpc}.
Recall that by Lemma~\ref{lem:classical-value-syntax}, if $Q \vdash v : P$ is a classical value, then $Q = \cdot.$ The proof proceeds by defining for every classical value $\cdot \vdash v : P$
a classical value interpretation  $ \flrb{\cdot \vdash v : P} : 1 \to \flrb P $ in $\Set.$ It is defined by induction on the derivation of $\cdot \vdash v : P$ in the following way:
\begin{align*}
  \flrb{\cdot \vdash * : I}                                          &\coloneqq \id_1 \\
  \flrb{\cdot \vdash \lleft_{P,R} v : P +R }                         &\coloneqq \mathrm{inl} \circ \flrb v \\
  \flrb{\cdot \vdash \rright_{P,R} v : P +R }                        &\coloneqq \mathrm{inr} \circ \flrb v \\
  \flrb{\cdot \vdash ( v, w ) : P \otimes R }                        &\coloneqq \left\langle \flrb v, \flrb w \right\rangle \\
  \flrb{\cdot \vdash \fold_{\mu X. P} v : \mu X. P}                  &\coloneqq \ifold \circ \flrb v ,
\end{align*}
where $\mathrm{inl}$ and $\mathrm{inr}$ are the coproduct injections in $\Set$;  $\langle f, g \rangle$ is the unique map induced by the product in $\Set$ and $\ifold$ is defined in Definition~\ref{def:fold-unfold-classical}.
To finish the proof, one has to show
\[ \lrb{\cdot \vdash v : P } = \left( I \xrightarrow{\iota^I} F1 \xrightarrow{F \flrb{\cdot \vdash v : P} } F \flrb{P} \xrightarrow{ (\iota^P)^{-1} } \lrb{P} \right) . \]
The $\fold$ case follows by Proposition~\ref{prop:fold-unfold-relationship} and the remaining cases follow easily from the axioms of symmetric monoidal adjunctions.
\qed\end{proof}

\newpage
\section{Proof of Theorem~\ref{thm:big-step-invariance}}\label{app:big-step}

For brevity, we define
\[ \lrb{\mathcal C \Downarrow} \coloneqq \bigvee_{n=0}^\infty \sum_{r \in \TerSeq_{\leq n}(\mathcal C)} \lrb{\End(r)} . \]
Note, that the sequence $\left( \sum_{r \in \TerSeq_{\leq n}(\mathcal C)} \lrb{\End(r)}\right)_{n \in \mathbb N}$ is an increasing chain in our order and thus we may take its supremum, as we did above.

To show Theorem~\ref{thm:big-step-invariance} is to show:
\[ \lrb{\mathcal C} = \lrb{\mathcal C \Downarrow} . \]

As a simple consequence of Soundness (Theorem~\ref{thm:soundness}), we get the following.

\begin{corollary}\label{cor:easy-part-adequacy}
For any configuration $\mathcal C$, we have
$ \lrb{\mathcal C} \geq \lrb{\mathcal C \Downarrow} . $
\end{corollary}
\begin{proof}
We have
$ \lrb{\mathcal C} \geq \sum_{r \in \TerSeq_{\leq n}(\mathcal C)} \lrb{\End(r)} $
which follows by induction on $n$ using Theorem~\ref{thm:soundness}. 
The corollary follows by taking the supremum over $n$.
\qed\end{proof}

In the remainder of this appendix, we will show that the converse inequality also holds, thereby finishing the proof of Theorem~\ref{thm:big-step-invariance}.



\subsection{Proof Strategy}

Our proof strategy is based on that of~\cite{quantitative} where the authors
establish computational adequacy at unit type.

In~\secref{sub:language-extension} we extend the QPL language with finitary (or
bounded) primitives for recursion and we update the operational and
denotational semantics in an appropriate way, so that the established language
properties continue to hold (Theorem~\ref{thm:extended-properties}).
In~\secref{sub:adequacy-proof} we prove adequacy via the following steps: we
show that any finitary configuration is strongly normalising
(Lemma~\ref{lem:finitary-strong-normalisation}) which then allows us to prove
invariance of the interpretation for finitary configurations with respect to
big-step reduction (Corollary~\ref{cor:finite-hard}); we show that the finitary
configurations approximate the ordinary configurations both
operationally (Lemma~\ref{lem:operational-approximation}) and
denotationally (Lemma~\ref{lem:denotational-approximation}); using these results we
prove invariance of the denotational semantics with respect to big-step
reduction for the ordinary QPL language (Theorem~\ref{thm:big-step-invariance}).

\subsection{Language Extension}\label{sub:language-extension}

We extend the syntax of QPL by adding \emph{indexed procedure names}. We use $f^n, g^n$ with $n \in \mathbb N$ to range over indexed procedure names.
An indexed procedure name $f^n$ can be used at most $n$ times in the operational semantics (see below), whereas an (ordinary) procedure name $f$ may be used any number of times.
We write $f^n : A \to B$ to indicate that the indexed procedure name has input type $A$ and output type $B$. Procedure contexts are now also allowed to contain indexed procedure names, in addition to standard procedure names.
Procedure contexts which contain a procedure name $f^n$ cannot contain the unindexed procedure name $f$ or any other other indexed procedure names $f^m$ with $m \neq n$.

The term language is extended by adding new terms, some of which are indexed by natural numbers $n \geq 0.$ These terms are governed by the following formation rules:

  \centerline{
  \begin{bprooftree}
  \AxiomC{\phantom{$\Pi \vdash \langle$}}
  \UnaryInfC{$\Pi \vdash \langle \Gamma \rangle\ \ZERO_{\Gamma, \Sigma}\ \langle  \Sigma \rangle$}
  \end{bprooftree}
  \begin{bprooftree}
  \AxiomC{$\Pi \vdash \langle \Gamma, b: \textbf{bit} \rangle\ M\ \langle  \Gamma, b: \textbf{bit} \rangle$}
  \UnaryInfC{$\Pi \vdash \langle \Gamma, b: \textbf{bit} \rangle\ 
    \textbf{while}^n\ b\ \textbf{do}\ M\ \langle \Gamma, b: \textbf{bit} \rangle$}
  \end{bprooftree}
  }
  \vspace{4mm}
  \centerline{
    \begin{bprooftree}
    \AxiomC{$\Pi, f^n: A \to B \vdash \langle x:A \rangle\ M\ \langle y:B \rangle$}
    \UnaryInfC{$\Pi \vdash \langle \Gamma \rangle\ \textbf{proc}\ f^n ::\ x:A \to y:B\ \{M\}\ \langle \Gamma \rangle$}
    \end{bprooftree}
    \begin{bprooftree}
    \AxiomC{\phantom{$f^n$}}
    \UnaryInfC{$\Pi, f^n: A \to B \vdash \langle \Gamma, x: A \rangle\ y = f^n(x)\ \langle \Gamma, y:B \rangle$}
    \end{bprooftree}
  }
  \vspace{4mm}
The formation rules for procedure stores are extended by adding a rule for indexed procedures:
  \[
    \begin{bprooftree}
    \AxiomC{$\Pi \vdash \Omega$}
    \AxiomC{$\Pi, f^n : A \to B \vdash \langle x:A \rangle\ M\ \langle y:B \rangle$}
    \BinaryInfC{$\Pi, f^n : A \to B \vdash  \Omega, f^n :: x:A \to y:B\ \{M\}$}
    \end{bprooftree}
  \]
The notion of a well-formed configuration is defined in the same way as before provided one uses the updated notions of well-formed terms and well-formed procedure stores.
The operational semantics is extended by adding the rules:

\vspace{4mm}
{
  \small{
  \centerline{
    \begin{bprooftree}
    \AxiomC{}
    \UnaryInfC{$ 
      (\ZERO ; P\ |\ V \ |\ \Omega\ |\ \rho)
      \leadsto
      (\ZERO\ |\ V \ |\ \Omega\ |\ \rho)
    $}
    \end{bprooftree}
    \begin{bprooftree}
    \AxiomC{}
    \UnaryInfC{$ 
      (\textbf{while}^{0}\ b\ \textbf{do}\ M\ |\ V, b = v\ |\ \Omega\ |\ \rho)
      \leadsto
      (\ZERO\ |\ V, b = v \ |\ \Omega\ |\ \rho)
    $}
    \end{bprooftree}
  }
  \vspace{4mm}
  \centerline{
    \begin{bprooftree}
    \AxiomC{}
    \UnaryInfC{$ 
      (\textbf{while}^{n+1}\ b\ \textbf{do}\ M\ |\ V, b = v\ |\ \Omega\ |\ \rho)
      \leadsto
      (\textbf{if}\ b\ \textbf{then}\ \{M; \textbf{while}^n\ b\ \textbf{do}\ M\}\ \textbf{else}\ \sskip\ |\ V, b = v \ |\ \Omega\ |\ \rho)
    $}
    \end{bprooftree}
  }
  \vspace{4mm}
  \centerline{
    \begin{bprooftree}
    \AxiomC{}
    \UnaryInfC{$
      (\textbf{proc}\ f^n::\ x:A \to y:B\ \{M\}\ |\ V\ |\ \Omega\ |\ \rho)
      \leadsto
      (\sskip\ |\ V\ |\ \Omega, f^n::\ x:A \to y:B\ \{M\}\ |\ \rho)
    $}
    \end{bprooftree}
  }
  \vspace{4mm}
  \centerline{
    \begin{bprooftree}
    \AxiomC{}
    \UnaryInfC{$
      (y_1 = f^{0}(x_1)\ |\ V, x_1 = v\ |\ \Omega, f^0::\ x_2:A \to y_2:B\ \{M\}\ |\ \rho)
      \leadsto
      (\ZERO\ |\ V, x_1 =v\ |\ \Omega, f^0::\ x_2:A \to y_2:B\ \{M\}\ |\ \rho)
    $}
    \end{bprooftree}
  }
  \vspace{4mm}
  }
  \footnotesize{
  \centerline{
    \begin{bprooftree}
    \AxiomC{}
    \UnaryInfC{$
      (y_1 = f^{n+1}(x_1)\ |\ V, x_1 = v\ |\ \Omega, f^{n+1}::\ x_2:A \to y_2:B\ \{M\}\ |\ \rho)
      \leadsto
      (M_{\alpha_1}\ |\ V, x_1 =v\ |\ \Omega_{\alpha}, f^n::\ x_2 : A \to y_2 : B\ \{M_{\alpha_2}\}\ |\ \rho)
    $}
    \end{bprooftree}
  }
  }
}
\vspace{4mm}

where in the last rule, $\Omega_\alpha$ is the procedure store obtained from $\Omega$ by renaming $f^{n+1}$ to $f^n$ within all procedure bodies contained in $\Omega$. Similarly, $M_{\alpha_2}$ is the result of renaming $f^{n+1}$ to $f^n$ within $M$.
The term $M_{\alpha_1}$ is obtained from $M$ by also renaming $f^{n+1}$ to $f^n$ and then doing the same term variable renamings as in the unindexed (call) rule (see \secref{sec:operational}).

We extend the denotational semantics by adding interpretations for the newly added terms:

{
\small
\centerline{
$
\lrb{\Pi \vdash \langle \Gamma \rangle\  \ZERO_{\Gamma, \Sigma}\ \langle \Sigma \rangle} \coloneqq \pi \mapsto \left(
      \lrb{\Gamma}
        \xrightarrow{\ZERO_{\lrb \Gamma, \lrb \Sigma}}
      \lrb{\Sigma}
      \right) 
$
}
\centerline{
$
\lrb{\Pi \vdash \langle \Gamma, b: \textbf{bit} \rangle\ \textbf{while}^n\ b\ \textbf{do}\ {M}\ \langle \Gamma, b: \textbf{bit} \rangle} \coloneqq \pi \mapsto \left(
      \lrb{\Gamma} \otimes \textbf{bit} 
        \xrightarrow{W^n_{\lrb M(\pi)}(\ZERO) }
      \lrb{\Gamma} \otimes \textbf{bit} \right)
$
}
\centerline{
$ \lrb{\Pi \vdash \langle \Gamma \rangle\ \textbf{proc}\ f^n ::\ x:A \to y:B\ \{M\}\ \langle \Gamma \rangle} \coloneqq \pi \mapsto\ \left( \lrb{\Gamma} \xrightarrow{\id} \lrb{\Gamma} \right) $
}
\centerline{
$ \lrb{\Pi, f^n: A \to B \vdash \langle \Gamma, x: A \rangle\ y = f^n(x)\ \langle \Gamma, y:B \rangle} \coloneqq (\pi, f) \mapsto \left( \lrb \Gamma \otimes \lrb A \xrightarrow{ \id \otimes f} \lrb \Gamma \otimes \lrb B \right) $
}
}

Notice that the interpretations of the indexed (proc) and (call) term rules
contain no fixpoint calculations, just like their unindexed counterparts.
Similarly to the unindexed case, the non-trivial calculations take place in the
interpretation of the procedure store rules. The newly added rule for formation
of procedure stores is interpreted in the following way:
\begin{align*}
\lrb{\Pi, f^n : A \to B \vdash \Omega, f^n :: x:A \to y:B\ \{M\}} = \left(\lrb \Omega, \left( \lrb{M}(\lrb \Omega, -) \right)^n(\ZERO) \right).
\end{align*}
The interpretation of configurations is now defined in the same way as before, but also using the newly added rules and updated notions.
Finally, the notion of terminal configuration is also updated to include configurations of the form $(\ZERO\ |\ V \ |\ \Omega\ |\ \rho)$.
With this in place, all language properties stated in the previous sections also hold true for the extended language.

\begin{theorem}\label{thm:extended-properties}
Subject Reduction (Theorem~\ref{thm:subject}), Progress (Theorem~\ref{thm:progress}) and Soundness (Theorem~\ref{thm:soundness}) also hold true for the extended language (using the updated language notions).
\end{theorem}

\subsection{The Proof}\label{sub:adequacy-proof}

A term in the extended language is called \emph{finitary} if it does not
contain any unindexed procedure names or while loops. A term is called
\emph{ordinary} if it does not contain any indexed procedure names, indexed while
loops or $\ZERO_{\Gamma, \Sigma}$ subterms. In other words, an ordinary term is simply a term in the base QPL
language as described in~\secref{sec:syntax}. Similarly, a \emph{finitary (ordinary) procedure store} $\Omega$ is a
procedure store where each procedure name of $\Omega$ is indexed
(unindexed) and such that its procedure body is a finitary (ordinary) term.  A
\emph{finitary (ordinary) configuration} is a configuration $(M\ |\ V \ |\ \Omega\ |\ \rho)$ where $M$ and $\Omega$ are finitary (ordinary).
A finitary configuration is true to its name, because all of its reduction sequences terminate in a finite number of steps.

\begin{lemma}(Finitary Strong Normalisation)\label{lem:finitary-strong-normalisation}
For any finitary configuration $\mathcal C$, there exists $n \in \mathbb N$, such that the length of every reduction sequence from $\mathcal C$ is at most $n$.
\end{lemma}


\begin{corollary}\label{cor:finite-hard}
For any finitary configuration $\mathcal C$, we have
$ \lrb{\mathcal C} = \lrb{\mathcal C \Downarrow} .$
\end{corollary}
\begin{proof}
Using Lemma~\ref{lem:finitary-strong-normalisation},
let $n \in \mathbb N$ be the smallest number such that the length of every reduction sequence from $\mathcal C$ is at most $n$. It follows that
\[  \lrb{\mathcal C} = \sum_{r \in \TerSeq_{\leq n} (\mathcal C)} \lrb{\End(r)}, \]
which can be easily shown by induction on $n$ using Theorem~\ref{thm:extended-properties}. The proof concludes by recognising that
\begin{align*}
\sum_{r \in \TerSeq_{\leq n} (\mathcal C)} \lrb{\End(r)} &= \bigvee_{n=0}^\infty \sum_{r \in \TerSeq_{\leq n} (\mathcal C)} \lrb{\End(r)} \\
& = \lrb{\mathcal C \Downarrow}
\end{align*}
\qed\end{proof}

%

Next, we define an approximation relation $(- \blacktriangleleft -)$ between finitary terms and ordinary terms. It is defined to be the smallest relation satisfying the following rules:

  \vspace{4mm}
  \centerline{
    \begin{bprooftree}
    \AxiomC{$M' \blacktriangleleft M$}
    \UnaryInfC{$\textbf{proc}\ f^n ::\ x:A \to y:B\ \{M'\} \blacktriangleleft \textbf{proc}\ f ::\ x:A \to y:B\ \{M\}$}
    \end{bprooftree}
  }
  \vspace{4mm}
  \centerline{
    \begin{bprooftree}
    \AxiomC{{\color{white}{M'}}}
    \UnaryInfC{$y = f^n(x) \blacktriangleleft y = f(x)$}
    \end{bprooftree}
    \begin{bprooftree}
    \AxiomC{$M' \blacktriangleleft M$}
    \UnaryInfC{$ \textbf{while}^n\ b\ \textbf{do}\ M' \blacktriangleleft \textbf{while}\ b\ \textbf{do}\ M $}
    \end{bprooftree}
    \begin{bprooftree}
    \AxiomC{$M' \blacktriangleleft M$}
    \AxiomC{$N' \blacktriangleleft N$}
    \BinaryInfC{$M';N' \blacktriangleleft M;N$}
    \end{bprooftree}
  }
  \vspace{4mm}
  \centerline{
    \begin{bprooftree}
    \AxiomC{$M_1' \blacktriangleleft M_1$}
    \AxiomC{$M_2' \blacktriangleleft M_2$}
    \BinaryInfC{$\ccase\ y\ \textbf{of}\ \{\lleft_{A,B}\ x_1 \to M_1'\ |\ \rright_{A,B}\ x_2 \to M_2'\ \} \blacktriangleleft \ccase\ y\ \textbf{of}\ \{\lleft_{A,B}\ x_1 \to M_1\ |\ \rright_{A,B}\ x_2 \to M_2\ \}$}
    \end{bprooftree}
  }
  \vspace{4mm}
  \centerline{
    \begin{bprooftree}
    \AxiomC{}
    \RightLabel{all other terms except $\ZERO.$} \UnaryInfC{$M \blacktriangleleft M$}
    \end{bprooftree}
  }
  \vspace{4mm}

We also define an approximation relation $(- \vartriangleleft -)$ between finitary procedure stores and ordinary procedure stores. It is defined to be the smallest relation satisfying the following rules:
\[
  \begin{bprooftree}
  \AxiomC{\phantom{$\Omega'$}}
  \UnaryInfC{$\cdot \vartriangleleft \cdot$}
  \end{bprooftree}
  \qquad
  \begin{bprooftree}
  \AxiomC{$\Omega' \vartriangleleft \Omega$}
  \AxiomC{$M' \blacktriangleleft M$}
  \BinaryInfC{$\Omega', f^n :: x:A \to y:B\ \{M'\} \vartriangleleft \Omega, f :: x:A \to y:B\ \{M\}$}
  \end{bprooftree}
\]
Finally, we define an approximation relation $(- \sqsubset - )$ between finitary configurations and ordinary configurations. It is defined to be the smallest relation satisfying the following rule:
\[
  \begin{bprooftree}
  \AxiomC{$M' \blacktriangleleft M$}
  \AxiomC{$\Omega' \vartriangleleft \Omega$}
  \BinaryInfC{$ (M'\ |\ V \ |\ \Omega'\ |\ \rho) \sqsubset (M\ |\ V \ |\ \Omega\ |\ \rho)$}
  \end{bprooftree}
\]

\begin{remark}
By definition, if $ (M'\ |\ V \ |\ \Omega'\ |\ \rho) \sqsubset (M\ |\ V \ |\ \Omega\ |\ \rho)$, then $M$ and $M'$ do not contain any $\ZERO$ subterms, nor do $\Omega$ and $\Omega'$ in any of their procedure bodies.
\end{remark}

We proceed with a lemma which shows the relation $( - \sqsubset -) $ approximates the ordinary operational semantics. But first, we introduce two new notions.
Every terminal finitary configuration $\mathcal T$ is either of the form $(\sskip\ |\ V \ |\ \Omega\ |\ \rho)$ or $(\ZERO\ |\ V \ |\ \Omega\ |\ \rho)$.
In the former case we say $\mathcal T$ is a $\sskip$-\emph{configuration} and in the latter case we say $\mathcal T$ is a $\ZERO$-\emph{configuration}.

\begin{lemma}\label{lem:operational-approximation}
Let $\mathcal C_0$ be an ordinary configuration and $C_0'$ a finitary configuration with $\mathcal C_0' \sqsubset \mathcal C_0$.
Let $r' = \left( \mathcal C_0' \leadsto \cdots \leadsto \mathcal C_n' \right)$ be a terminal reduction sequence, where $\mathcal C_n'$ is a $\sskip$-configuration.
Then there exists a unique terminal reduction sequence $r= \left( \mathcal C_0 \leadsto \cdots \leadsto \mathcal C_n \right)$ such that $\mathcal C_i' \sqsubset \mathcal C_i$,
which we will henceforth abbreviate by writing $r' \sqsubset r.$
\end{lemma}
\begin{proof}
Let $\mathcal C_0' = (M'\ |\ V'\ |\ \Omega'\ |\ \rho')$. The proof follows by induction on $n$. The base case $n=0$ is trivial and the step case follows by case distinction on $M'$.
\qed\end{proof}

In other words, the above lemma shows that lockstep execution occurs between any ordinary
reduction sequence and any one of its approximating finitary reduction
sequences, provided the latter terminates in the ordinary sense.
This allows us to establish the following corollary.

\begin{corollary}\label{cor:inequal}
Let $\mathcal C$ and $\mathcal C'$ be two configurations with $\mathcal C' \sqsubset \mathcal C.$
Then
\[ \lrb{\mathcal C' \Downarrow} \leq \lrb{\mathcal C \Downarrow}. \]
\end{corollary}
\begin{proof}
It suffices to show for any $n \in \mathbb N$ that
\[ \sum_{r' \in \TerSeq_{\leq n} (\mathcal C')} \lrb{\End(r')} \leq \sum_{r \in \TerSeq_{\leq n} (\mathcal C)} \lrb{\End(r)} \]
from which the corollary follows by taking suprema.
Let $r' \in \TerSeq_{\leq n} (\mathcal C')$ be arbitrary. If $\End(r')$ is a $\ZERO$-configuration, then $\lrb{\End(r')} = \ZERO$ and thus it contributes nothing to the sum and may be ignored.
Otherwise $\End(r')$ is a $\sskip$-configuration and by Lemma~\ref{lem:operational-approximation}, there exists a unique $r \in \TerSeq_{\leq n} (\mathcal C)$, such that $r' \sqsubset r.$ In
this case $\End(r') \sqsubset \End(r)$ and since both of them are $\sskip$-configurations, it follows $\lrb{\End(r')} = \lrb{\End(r)}$. Moreover, if $r_1', r_2' \in \TerSeq_{\leq n} (\mathcal C')$ and also $r_1 \sqsupset r_1' \neq r_2' \sqsubset r_2$,
then $r_1 \neq r_2$
(to see this, observe that $r_1'$ and $r_2'$ must differ due to branching from a quantum measurement and thus so do $r_1$ and $r_2$).
Both sums are finite and we showed that for all non-trivial summands on the left there exist corresponding summands on the right which are equal to them, thus the required inequality holds.
\qed\end{proof}

The next lemma shows that ordinary configurations are also approximated by
finitary configurations in a denotational sense.

\begin{lemma}\label{lem:denotational-approximation}
For any ordinary term $M$, procedure store $\Omega$ and configuration $\mathcal C$:
\begin{align*}
\lrb{M} &= \bigvee_{M' \blacktriangleleft M} \lrb{M'} &
\lrb{\Omega} &= \bigvee_{\Omega' \vartriangleleft \Omega} \lrb{\Omega'} &
\lrb{\mathcal C} &= \bigvee_{\mathcal C' \sqsubset \mathcal C} \lrb{\mathcal C'} .
\end{align*}
\end{lemma}
\begin{proof}
Straightforward induction.
\qed\end{proof}

Finally, we can prove Theorem~\ref{thm:big-step-invariance}.

\emph{Proof of Theorem~\ref{thm:big-step-invariance}.}
By Corollary~\ref{cor:easy-part-adequacy}, it suffices to show $\lrb{\mathcal C} \leq \lrb{ \mathcal C \Downarrow }$. We have:
\begin{align*}
\lrb{\mathcal C} &= \bigvee_{\mathcal C' \sqsubset \mathcal C} \lrb{\mathcal C'} & (\text{Lemma~\ref{lem:denotational-approximation}}) \\
                 &= \bigvee_{\mathcal C' \sqsubset \mathcal C} \lrb{\mathcal C' \Downarrow} & (\text{Corollary~\ref{cor:finite-hard}}) \\
                 &\leq \lrb{\mathcal C \Downarrow} &(\text{Corollary~\ref{cor:inequal}})
\end{align*}
\qed

\newpage
\section{Proof of Theorem~\ref{thm:adequacy}}\label{app:adequacy}

Before we prove adequacy, we need a simple lemma.

\begin{lemma}\label{lem:terminal-trace}
Let $\mathcal T = (\sskip\ |\ V\ |\ \Omega\ |\ \rho)$ be a terminal configuration. Then
\[ \left( \diamond \circ \lrb{\mathcal T} \right)(1) = \trace \rho . \]
\end{lemma}
\begin{proof}
\begin{align*}
   \left( \diamond \circ \lrb{\mathcal T } \right) (1) 
&= \left( \diamond \circ \id \circ \lrb V \circ \mathrm{new}_\rho^\ddagger \right)(1) &(\text{Definition}) \\
&= \left( \mathrm{tr}^\ddagger \circ \mathrm{new}_\rho^\ddagger \right)(1) &(\text{Causality of values (see~\secref{sub:configurations})}) \\
&= \mathrm{new}_\rho^\dagger(\mathrm{tr}^\dagger(1)) &(\text{Definition})\\
&= \trace \rho &(\text{Definition})
\end{align*}
\qed\end{proof}

We can now prove Theorem~\ref{thm:adequacy}.

\emph{Proof of Theorem~\ref{thm:adequacy}.}
We have:
\begin{align*}
   \left( \diamond \circ \lrb{\mathcal C} \right)(1)
&= \left( \diamond \circ \bigvee_{n=0}^\infty \inlinesum_{r \in \TerSeq_{\leq n}(\mathcal C)} \lrb{\End(r)} \right)(1) &(\text{Theorem~\ref{thm:big-step-invariance}}) \\
&= \bigvee_{n=0}^\infty \inlinesum_{r \in \TerSeq_{\leq n}(\mathcal C)} \left(\diamond \circ \lrb{\End(r)} \right)(1) &\text{} \\
&= \bigvee_{n=0}^\infty \inlinesum_{r \in \TerSeq_{\leq n}(\mathcal C)} \trace{\End(r)} &(\text{Lemma~\ref{lem:terminal-trace}}) \\
&= \Halt(\mathcal C ) &(\text{Definition})
\end{align*}
\qed

\end{document}